\documentclass[12pt]{article}
\usepackage{amsxtra}
\usepackage{amssymb, amsmath}
\usepackage{amsthm}
\usepackage {hyperref}

\newtheorem{lemma}{Lemma}[section]
\newtheorem{theorem}{Theorem}[section]
\newtheorem{remark}{Remark}[section]

\newtheorem{proposition}{Proposition}[section]

\newtheorem{property}{Property}[section]
\usepackage[dvips]{graphicx}

\begin{document}
\begin{center}
\textbf{\LARGE{A Sequence of Inequalities among Difference of Symmetric Divergence Measures}}
\end{center}

\smallskip
\begin{center}
\textbf{\large{Inder Jeet Taneja}}\\
Departamento de Matem\'{a}tica\\
Universidade Federal de Santa Catarina\\
88.040-900 Florian\'{o}polis, SC, Brazil.\\
\textit{e-mail: taneja@mtm.ufsc.br\\
http://www.mtm.ufsc.br/$\sim $taneja}
\end{center}

\begin{abstract}
In this paper we have considered two one parametric generalizations. These two generalizations have in particular the well known measures such as: \textit{J-divergence}, \textit{Jensen-Shannon divergence} and  \textit{ arithmetic-geometric mean divergence. }These three measures are with logarithmic expressions. Also, we have particular cases the measures such as: \textit{Hellinger discrimination}, \textit{symmetric }$\chi ^2 - $\textit{divergence}, and \textit{triangular discrimination}. These three measures are also well-known in the literature of statistics, and are without
logarithmic expressions. Still, we have one more non logarithmic measure as particular case calling it \textit{d-divergence}. These seven measures bear an interesting inequality. Based on this inequality, we have considered different \textit{difference of divergence measures} and established a sequence of inequalities among themselves.
\end{abstract}



\section{Introduction}

Let
\[
\Gamma _n = \left\{ {P = (p_1 ,p_2 ,...,p_n )\left| {p_i > 0,\sum\limits_{i
= 1}^n {p_i = 1} } \right.} \right\}, \quad n \ge 2,
\]

\noindent
be the set of all complete finite discrete probability distributions. For
all $P,Q \in \Gamma _n $, the following measures are well known in the
literature on information theory and statistics:

\bigskip
\noindent
\textbf{$\bullet$ Hellinger Discrimination}

\[
h(P\vert \vert Q) = 1 - B(P\vert \vert Q) = \frac{1}{2}\sum\limits_{i = 1}^n
{(\sqrt {p_i } - \sqrt {q_i } )^2} ,
\]

\noindent where

\[
B(P\vert \vert Q) = \sqrt {p_i q_i } ,
\]

\noindent
is the well-known Bhattacharyya \textit{coefficient}.

\bigskip
\noindent
\textbf{$\bullet$ Triangular Discrimination}

\[
\Delta (P\vert \vert Q) = 2\left[ {1 - W(P\vert \vert Q)} \right] =
\sum\limits_{i = 1}^n {\frac{(p_i - q_i )^2}{p_i + q_i }} ,
\]

\noindent where

\[
W(P\vert \vert Q) = \sum\limits_{i = 1}^n {\frac{2p_i q_i }{p_i + q_i }} ,
\]

\noindent
is the well-known \textit{harmonic mean divergence}.

\bigskip
\noindent
\textbf{$\bullet$ Symmetric Chi-square Divergence}

\[
\Psi (P\vert \vert Q) = \chi ^2(P\vert \vert Q) + \chi ^2(Q\vert \vert P) =
\sum\limits_{i = 1}^n {\frac{(p_i - q_i )^2(p_i + q_i )}{p_i q_i }} ,
\]

\noindent where

\[
\chi ^2(P\vert \vert Q) = \sum\limits_{i = 1}^n {\frac{(p_i - q_i )^2}{q_i
}} = \sum\limits_{i = 1}^n {\frac{p_i^2 }{q_i } - 1} ,
\]

\noindent
is the well-known $\chi ^2 - $\textit{divergence} .

\bigskip
\noindent
\textbf{$\bullet$ J-Divergence}

\[
J(P\vert \vert Q) = \sum\limits_{i = 1}^n {(p_i - q_i )\ln (\frac{p_i }{q_i
})} .
\]

\bigskip
\noindent
\textbf{$\bullet$ Jensen-Shannon Divergence}

\[
I(P\vert \vert Q) = \frac{1}{2}\left[ {\sum\limits_{i = 1}^n {p_i \ln \left(
{\frac{2p_i }{p_i + q_i }} \right) + } \sum\limits_{i = 1}^n {q_i \ln \left(
{\frac{2q_i }{p_i + q_i }} \right)} } \right].
\]

\bigskip
\noindent
\textbf{$\bullet$ Arithmetic-Geometric Mean Divergence}

\[
T(P\vert \vert Q) = \sum\limits_{i = 1}^n {\left( {\frac{p_i + q_i }{2}}
\right)\ln \left( {\frac{p_i + q_i }{2\sqrt {p_i q_i } }} \right)} .
\]

Originally, the \textit{J-divergence} is due to Jeffreys \cite{jef}. The measure \textit{Jensen-Shannon divergence} is due to Sibson \cite{sib}. Later, Burbea and Rao \cite{bur} studied it extensively. The \textit{arithmetic and geometric mean divergence} is due to Taneja \cite{tan2}. Detailed study of these measures can be seen in Taneja \cite{tan2, tan3, tan4}. We call the above six measures \textit{symmetric divergence measures}, since they are symmetric with respect to the probability distributions $P$ and $Q$. The author \cite{tan3, tan4} obtained an inequality among these six symmetric divergence measures given by
\begin{equation}
\label{eq1}
\frac{1}{4}\Delta (P\vert \vert Q) \le I(P\vert \vert Q) \le h(P\vert \vert
Q)
 \le \frac{1}{8}J(P\vert \vert Q) \le T(P\vert \vert Q) \le \frac{1}{16}\Psi
(P\vert \vert Q).
\end{equation}

By defining the nonnegative differences among the divergence measures
appearing in (\ref{eq1}), the author \cite{tan4} improved the above result
(\ref{eq1}) obtaining the following sequence of inequalities:
\begin{align}
\label{eq2}
 D_{I\Delta } (P\vert \vert Q) & \le \frac{2}{3} D_{h\Delta } (P\vert \vert Q)
\le \frac{1}{2}D_{J\Delta } (P\vert \vert Q) \le \frac{1}{3}D_{T\Delta }
(P\vert \vert Q) \le D_{TJ} (P\vert \vert Q) \le\notag\\
 \le &\frac{2}{3}D_{Th} (P\vert \vert Q) \le 2D_{Jh} (P\vert \vert Q) \le
\frac{1}{6}D_{\Psi \Delta } (P\vert \vert Q) \le \frac{1}{5}D_{\Psi I}
(P\vert \vert Q) \le\notag\\
 &\le \frac{2}{9}D_{\Psi h} (P\vert \vert Q) \le \frac{1}{4}D_{\Psi J}
(P\vert \vert Q) \le \frac{1}{3}D_{\Psi T} (P\vert \vert Q),
\end{align}

\noindent
where, for example, $D_{\Psi T} (P\vert \vert Q) = \frac{1}{16}\Psi (P\vert
\vert Q) - T(P\vert \vert Q)$, and similarly others. Still, we have
\begin{equation}
\label{eq3}
\frac{2}{3}D_{h\Delta } (P\vert \vert Q) \le 2D_{hI} (P\vert \vert Q) \le
D_{TJ} (P\vert \vert Q).
\end{equation}

The proof of the inequalities (\ref{eq1})-(\ref{eq3}) is based on the following two lemmas:

\begin{lemma} If the function $f:[0,\infty ) \to {\rm R}$ is convex and normalized, i.e., $f(1) = 0$, then the \textit{f-divergence}, $C_f (P\vert \vert Q)$ given by
\begin{equation}
\label{eq4}
C_f (P\vert \vert Q) =
\sum\limits_{i = 1}^n {q_i f\left( {\frac{p_i }{q_i }} \right)} ,
\end{equation}

\noindent
is nonnegative and convex in the pair of probability distribution $(P,Q) \in
\Gamma _n \times \Gamma _n $.
\end{lemma}

\begin{lemma} Let $f_1 ,f_2 :I \subset {\rm R}_ + \to {\rm R}$ two
generating mappings are normalized, i.e., $f_1 (1) = f_2 (1) = 0$ and
satisfy the assumptions:

(i) $f_1 $ and $f_2 $ are twice differentiable on $(a,b)$;

(ii) there exists the real constants $m,M$such that $m < M$and

\[
m \le \frac{f_1 ^{\prime \prime }(x)}{f_2 ^{\prime \prime }(x)} \le M,
\quad f_2 ^{\prime \prime }(x) > 0, \quad \forall x \in (a,b),
\]

\noindent then we have the inequalities:
\begin{equation}
\label{eq5}
m\mbox{ }C_{f_2 } (P\vert \vert Q) \le C_{f_1 } (P\vert \vert Q) \le M\mbox{
}C_{f_2 } (P\vert \vert Q).
\end{equation}
\end{lemma}

 The measure (\ref{eq5}) is the well-known \textit{Csisz\'{a}r's f-divergence}. The Lemma 1.1 is due to Csisz\'{a}r \cite{csi} and the Lemma 1.2 is due to author \cite{tan4}. Some applications of Lemma 1.1 can be seen in Taneja and Kumar \cite{tak}.

\section{Generalized Symmetric Divergence Measures}

Let us consider the measure
\begin{equation}
\label{eq6}
\zeta _s (P\vert \vert Q) = \begin{cases}
 {J_s (P\vert \vert Q) = \left[ {s(s - 1)} \right]^{ - 1}\left[
{\sum\limits_{i = 1}^n {\left( {p_i^s q_i^{1 - s} + p_i^{1 - s} q_i^s }
\right) - 2} } \right],} & {s \ne 0,1} \\
 {J(P\vert \vert Q) = \sum\limits_{i = 1}^n {\left( {p_i - q_i } \right)\ln
\left( {\frac{p_i }{q_i }} \right),} } & {s = 0,1} \\
\end{cases}
\end{equation}

\noindent
for all $P,Q \in \Gamma _n $

\bigskip
The measure (\ref{eq6}) is \textit{generalized J-divergence} or \textit{J-divergence of type s} and is extensively studied in Taneja \cite{tan2, tan4}. The expression (\ref{eq6}) admits the following particular cases:

\smallskip
\noindent
(i) $\zeta _{ - 1} (P\vert \vert Q) = \zeta _2 (P\vert \vert Q) = \frac{1}{2}\Psi (P\vert \vert Q)$,

\noindent
(ii) $\zeta _0 (P\vert \vert Q) = \zeta _1 (P\vert \vert Q) = J(P\vert \vert Q)$,

\noindent
(iii) $\zeta _{1 / 2} (P\vert \vert Q) = 8\,h(P\vert \vert Q)$,

\smallskip
\noindent
where $\Psi (P\vert \vert Q)$, $J(P\vert \vert Q)$and $h(P\vert \vert Q)$are
as given in section 1.

\bigskip
Let us consider now the another measure
\begin{equation}
\label{eq7}
\xi _s (P\vert \vert Q) =
\begin{cases}
 {IT_s (P\vert \vert Q) = \left[ {s(s - 1)} \right]^{ - 1}\left[
{\sum\limits_{i = 1}^n {\left( {\frac{p_i^s + q_i^s }{2}} \right)\left(
{\frac{p_i + q_i }{2}} \right)} ^{1 - s} - 1} \right],} & {s \ne 0,1} \\
 {I(P\vert \vert Q) = \frac{1}{2}\left[ {\sum\limits_{i = 1}^n {p_i \ln
\left( {\frac{2p_i }{p_i + q_i }} \right) + \sum\limits_{i = 1}^n {q_i \ln
\left( {\frac{2q_i }{p_i + q_i }} \right)} } } \right],} & {s = 1} \\
 {T(P\vert \vert Q) = \sum\limits_{i = 1}^n {\left( {\frac{p_i + q_i }{2}}
\right)\ln \left( {\frac{p_i + q_i }{2\sqrt {p_i q_i } }} \right)} ,} & {s =
0} \\
\end{cases}
\end{equation}

\noindent
for all $P,Q \in \Gamma _n $

\bigskip
The measure (\ref{eq7}) is new in the literature and is studied for the first time by Taneja \cite{tan3}. It is called \textit{generalized arithmetic and geometric mean divergence measure.} . The measure (\ref{eq7}) admits the following particular cases:

\smallskip
\noindent
(i) $\xi _{ - 1} (P\vert \vert Q) = \frac{1}{4}\Delta (P\vert \vert Q)$.

\noindent
(ii) $\xi _1 (P\vert \vert Q) = I(P\vert \vert Q)$.

\noindent
(iii) $\xi _{1 / 2} (P\vert \vert Q) = 4\;d(P\vert \vert Q)$.

\noindent
(iii) $\xi _0 (P\vert \vert Q) = T(P\vert \vert Q)$.

\noindent
(iv) $\xi _2 (P\vert \vert Q) = \frac{1}{16}\Psi (P\vert \vert Q)$.

\smallskip
\noindent
where $\Delta (P\vert \vert Q)$, $I(P\vert \vert Q)$, $T(P\vert \vert Q)$
and $\Psi (P\vert \vert Q)$ are as given in section 1. The measure $d(P\vert
\vert Q)$ appearing in particular case (iii) is given by
\begin{equation}
\label{eq8}
d(P\vert \vert Q) = 1 - \sum\limits_{i = 1}^n {\left( {\frac{\sqrt {p_i } +
\sqrt {q_i } }{2}} \right)} \left( {\sqrt {\frac{p_i + q_i }{2}} } \right).
\end{equation}

For simplicity, we call the measure (\ref{eq8}) as $d - $\textit{divergence}. Thus we observe that when we take $s = \, - 1,\,0,\,\textstyle{1 \over 2},\,1$ and $2$ in (\ref{eq6}) and (\ref{eq7}), we have \textit{seven} particular cases. The measure $\Psi (P\vert \vert Q)$ appears as a particular case in both the measures (\ref{eq6}) and (\ref{eq7}). An inequality among these seven measures is given by
\begin{equation}
\label{eq9}
\frac{1}{4}\Delta (P\vert \vert Q) \le I(P\vert \vert Q) \le h(P\vert \vert
Q) \le 4\,d(P\vert \vert Q)
 \le \frac{1}{8}J(P\vert \vert Q) \le T(P\vert \vert Q) \le \frac{1}{16}\Psi
(P\vert \vert Q).
\end{equation}

Results appearing in (\ref{eq2}) and (\ref{eq3}) are based on the inequalities given in (\ref{eq1}). In this paper our aim is improve the results given in (\ref{eq2})-(\ref{eq3}) and obtain a new sequence of inequalities based on the expression (\ref{eq9}).

\section{Difference of Divergence Measures and their Convexity}

The inequality (\ref{eq9}) admits many nonnegative difference than the one given
(\ref{eq2}) and (\ref{eq3}), but we shall consider only those having the expression
$d(P\vert \vert Q)$. These \textit{nonnegative} differences are given by
\begin{align}
D_{\Psi d} (P\vert \vert Q) &= \frac{1}{16}\Psi (P\vert \vert Q) - 4d(P\vert \vert Q),\notag\\
D_{Td} (P\vert \vert Q) &= T(P\vert \vert Q) - 4d(P\vert \vert Q),\notag\\
D_{Jd} (P\vert \vert Q)& = \frac{1}{8}J(P\vert \vert Q) - 4d(P\vert \vert Q), \notag\\
D_{dh} (P\vert \vert Q) &= 4d(P\vert \vert Q) - h(P\vert \vert Q),\notag\\
D_{dI} (P\vert \vert Q)& = 4d(P\vert \vert Q) - I(P\vert \vert Q), \notag\\
\intertext{and}
D_{d\Delta } (P\vert \vert Q) &= 4d(P\vert \vert Q) - \frac{1}{4}\Delta (P\vert \vert Q).\notag
\end{align}

Here below we shall prove the convexity of the above \textit{six }measures. The proof is based on the Lemma 1.1. Initially, we shall give the convexity of the measures (\ref{eq6}) and (\ref{eq7}).

\begin{property} (i) The measure $\zeta _s (P\vert \vert Q)$ is \textit{nonnegative} and \textit{convex} in the pair of probability distributions $(P,Q) \in \Gamma _n \times \Gamma _n $ for all $s \in ( - \infty ,\infty )$.

\noindent
(ii) The measure $\xi _s (P\vert \vert Q)$ is \textit{nonnegative} and \textit{convex} in the pair of probability distributions $(P,Q) \in \Gamma _n \times \Gamma _n $ for all $s
\in ( - \infty ,\infty )$.
\end{property}

\begin{proof} (i) For all $x > 0$ and $s \in ( - \infty ,\infty )$, let us consider
\[
\phi _s (x) = \begin{cases}
 {\left[ {s(s - 1)} \right]^{ - 1}\left[ {x^s + x^{1 - s} - (1 + x)}
\right],} & {s \ne 0,1,} \\
 {(x - 1)\ln x,} & {s = 0,1,} \\
\end{cases}
\]

\noindent
in (\ref{eq4}), then we have $C_f (P\vert \vert Q)=\zeta _s \left( {P\vert
\vert Q} \right)$, where $\zeta _s \left( {P\vert \vert Q} \right)$ is given
by (\ref{eq6}).

\smallskip
Moreover,
\[
\phi _s ^\prime (x) = \begin{cases}
 {\left[ {s(s - 1)} \right]^{ - 1}\left[ {s(x^{s - 1} + x^{ - s}) + x^{ - s}
- 1} \right],} & {s \ne 0,1} \\
 {1 - x^{ - 1} + \ln x,} & {s = 0,1} \\
\end{cases},
\]

\noindent and
\begin{equation}
\label{eq10}
\phi _s ^{\prime \prime }(x) = x^{s - 2} + x^{ - s - 1}.
\end{equation}

Thus we have $\phi _s ^{\prime \prime }(x) > 0$ for all $x > 0$, and hence, $\phi _s (x)$ is convex for all $x > 0$. Also, we have $\phi _s (1) = 0$. In view of this we can say that the \textit{J-divergence of type s }is \textit{nonnegative} and \textit{convex} in the pair of probability distributions $(P,Q) \in \Gamma _n \times \Gamma _n $ for any $s \in ( - \infty ,\infty )$.

\smallskip
For all $x > 0$ and $s \in ( - \infty ,\infty )$, let us consider
\[
\psi _s (x) = \begin{cases}
 {\left[ {s(s - 1)} \right]^{ - 1}\left[ {\left( {\frac{x^{1 - s} + 1}{2}}
\right)\left( {\frac{x + 1}{2}} \right)^s - \left( {\frac{x + 1}{2}}
\right)} \right],} & {s \ne 0,1,} \\
 {\frac{x}{2}\ln x - \left( {\frac{x + 1}{2}} \right)\ln \left( {\frac{x +
1}{2}} \right),} & {s = 0,} \\
 {\left( {\frac{x + 1}{2}} \right)\ln \left( {\frac{x + 1}{2\sqrt x }}
\right),} & {s = 1,} \\
\end{cases}
\]

\noindent
in (\ref{eq4}), then we have $C_f (P\vert \vert Q) = \xi _s (P\vert \vert Q)$,
where $\xi _s (P\vert \vert Q)$ is as given by (\ref{eq7}).

\smallskip
Moreover,
\[
\psi _s ^\prime (x) = \begin{cases}
 {(s - 1)^{ - 1}\left[ {\frac{1}{s}\left[ {\left( {\frac{x + 1}{2x}}
\right)^s - 1} \right] - \frac{x^{ - s} - 1}{4}\left( {\frac{x + 1}{2}}
\right)^{s - 1}} \right],} & {s \ne 0,1,} \\
 { - \frac{1}{2}\ln \left( {\frac{x + 1}{2x}} \right),} & {s = 0,} \\
 {1 - x^{ - 1} - \ln x - 2\ln \left( {\frac{2}{x + 1}} \right),} & {s = 1,}
\\
\end{cases}
\]

\noindent and
\begin{equation}
\label{eq11}
\psi _s ^{\prime \prime }(x) = \left( {\frac{x^{ - s - 1} + 1}{8}}
\right)\left( {\frac{x + 1}{2}} \right)^{s - 2}.
\end{equation}

Thus we have $\psi _s ^{\prime \prime }(x) > 0$ for all $x > 0$, and hence, $\psi _s (x)$ is convex for all $x > 0$. Also, we have $\psi _s (1) = 0$. In view of this we can say that\textit{ AG and JS -- divergence of type }$s$ is \textit{nonnegative} and \textit{convex} in the pair of probability distributions $(P,Q) \in \Gamma _n \times \Gamma _n $ for any $s \in ( - \infty ,\infty )$.
\end{proof}

\begin{property} (i) The measure $\zeta _s (P\vert \vert Q)$ is monotonically increasing in $s$ for all $s \ge \frac{1}{2}$ and decreasing in $s \le \frac{1}{2}$.

\noindent
(ii)The measure $\xi _s (P\vert \vert Q)$ is monotonically increasing in $s$ for all $s \ge - 1$.
\end{property}

\smallskip
The proof of the Properties 2.1 and 2.2 can be seen in Taneja \cite{tan3}. Here, we have repeated the proof of (\ref{eq6}), since we need the expressions given in (\ref{eq10}) and (\ref{eq11}).

\begin{lemma} The above six difference of divergence measures are \textit{convex} in the pair of probability distributions $(P,Q) \in \Gamma _n \times \Gamma _n $.
\end{lemma}

\begin{proof}
We shall prove the above lemma in each case separately.

\bigskip
\noindent
(i) \textbf{For }$D_{\Psi d} (P\vert \vert Q)$\textbf{: }We can write
\[
D_{\Psi d} (P\vert \vert Q) = \frac{1}{16}\Psi (P\vert \vert Q) - 4d(P\vert
\vert Q)
 = \sum\limits_{i = 1}^n {q_i f_{\Psi d} \left( {\frac{p_i }{q_i }} \right)} ,
\]

\noindent where
\[
f_{\Psi d} (x) = \frac{1}{16}f_\Psi \left( x \right) - 4f_d \left( x
\right), \quad x > 0.
\]

We have
\begin{align}
\label{eq12}
{f}''_{\Psi d} (x) & = \frac{1}{16}{f}''_\Psi \left( x \right) - 4{f}''_d
\left( x \right)\notag\\
& = \frac{x^3 + 1}{8x^3} - \frac{x^{3 / 2} + 1}{2\sqrt 2 \;x^{3 / 2}(x +
1)^{3 / 2}}\notag\\
& = \frac{(x^3 + 1)(x + 1)\sqrt {2x + 2} - 4x^{3 / 2}(x^{3 / 2} +
1)}{8x^3\sqrt {2x + 2} }\notag\\
& = \frac{1}{x^3(x + 1)\sqrt {2x + 2} }\times m_1 (x),
\end{align}

\noindent
where $m_1 (x),\;x > 0$ is given by
\[
m_1 (x) = \left( {\frac{x^3 + 1}{2}} \right)\left( {\frac{x + 1}{2}}
\right)^{3 / 2} - x^{3 / 2}\left( {\frac{x^{3 / 2} + 1}{2}} \right).
\]

The measures ${f}''_\Psi \left( x \right)$ and ${f}''_d \left( x \right)$ appearing in (\ref{eq12}) are obtained from (\ref{eq11}) by taking $s = 2$ and $s = 1$ respectively. The graph of the function $m_1 (x),\;x > 0$ is given by

\begin{figure}[htbp]
\centerline{\includegraphics[width=2.00in,height=2.00in]{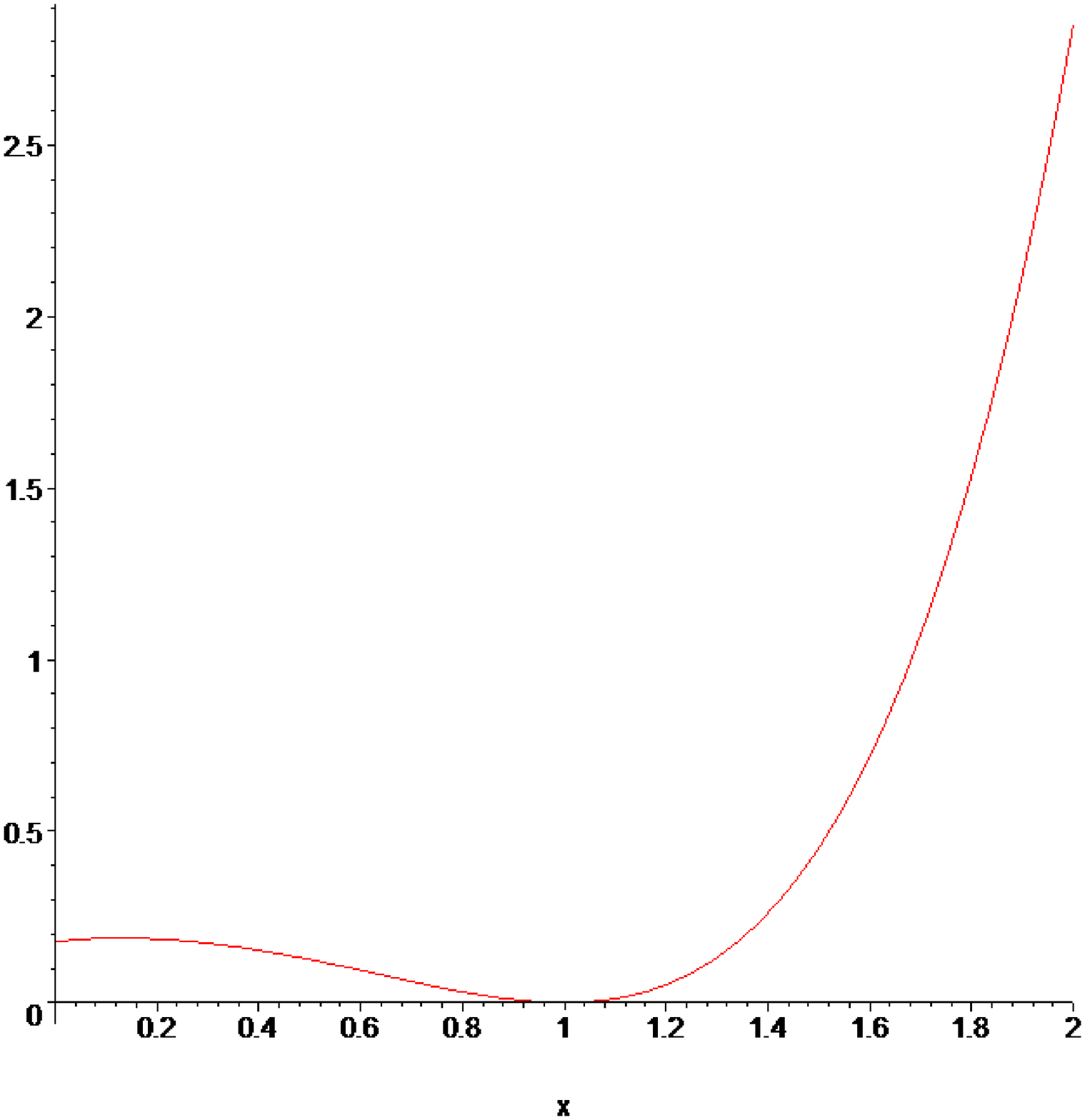}}
\label{fig1}
\end{figure}

From the above graph we observe that the function $m_1 (x) \ge 0,\;\forall x > 0$. This allows us to conclude that ${f}''_{\Psi d} (x) \ge 0,\;x > 0$, and hence, $f_{\Psi d} (x)$ is convex for all $x > 0$. Also, $f_{\Psi d} (1) = 0$. Thus by the application of the Lemma 1.1, we conclude that measure
$D_{\Psi d} (P\vert \vert Q)$is nonnegative and convex for all $(P,Q) \in \Gamma _n \times \Gamma _n $.

\bigskip
\noindent
(ii) \textbf{For }$D_{Td} (P\vert \vert Q)$\textbf{: }We can write
\[
D_{Td} (P\vert \vert Q) = T(P\vert \vert Q) - 4d(P\vert \vert Q)
 = \sum\limits_{i = 1}^n {q_i f_{Td} \left( {\frac{p_i }{q_i }} \right)} ,
\]

\noindent where
\[
f_{Td} (x) = f_T \left( x \right) - 4f_d \left( x \right),
\quad
x > 0.
\]

We have

\begin{align}
\label{eq13}
{f}''_{Td} (x) & = {f}''_T \left( x \right) - 4{f}''_d \left( x \right) \notag\\
& = \frac{1 + x^2}{4x^2(x + 1)} - \frac{x^{3 / 2} + 1}{2\sqrt 2 \;x^{3 / 2}(x
+ 1)^{3 / 2}}\notag\\
 &= \frac{1}{x^2(x + 1)\sqrt {2x + 2} }\times m_2 (x),
\end{align}

\noindent
where $m_2 (x),\;x > 0$ is given by
\[
m_2 (x) = \left( {\frac{x^2 + 1}{2}} \right)\left( {\frac{x + 1}{2}}
\right)^{1 / 2} - \sqrt x \left( {\frac{x^{3 / 2} + 1}{2}} \right).
\]

The measures ${f}''_T \left( x \right)$ and ${f}''_d \left( x \right)$ appearing in (\ref{eq13}) are obtained from (\ref{eq11}) and by taking $s = 1$ and $s = \frac{1}{2}$. (dividing by 4) respectively. The graph of the function $m_2 (x),\;x > 0$ is given by

\begin{figure}[htbp]
\centerline{\includegraphics[width=2.00in,height=2.00in]{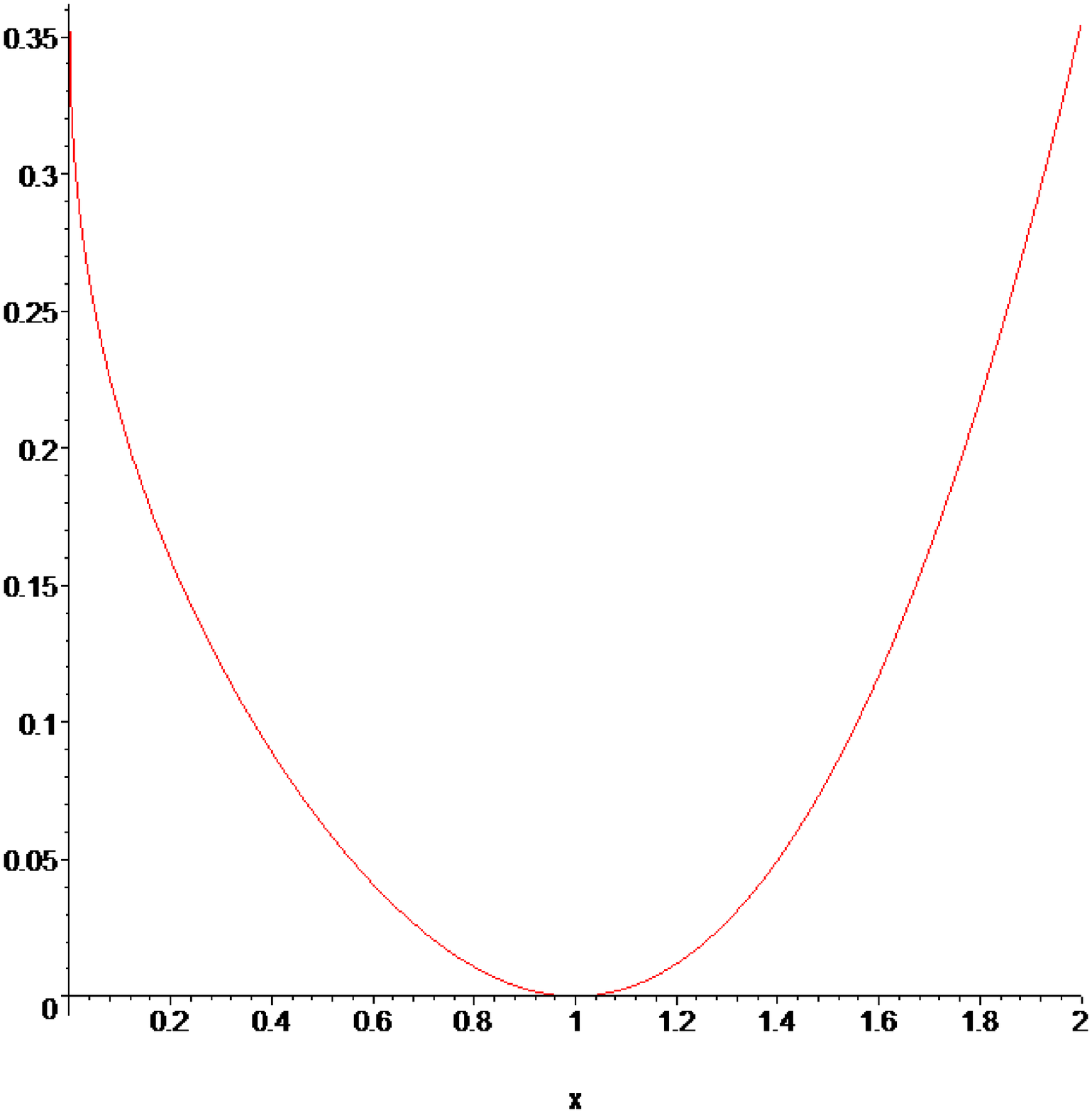}}
\label{fig2}
\end{figure}

From the above graph we observe that the function $m_2 (x) \ge 0,\;\forall x
> 0$. This allows us to conclude that ${f}''_{Td} (x) \ge 0,\;x > 0$, and
hence, $f_{Td} (x)$ is convex for all $x > 0$. Also $f_{Td} (1) = 0$. Thus
by the application of the Lemma 1.1, we conclude that measure $D_{Td}
(P\vert \vert Q)$ is nonnegative and convex for all $(P,Q) \in \Gamma _n
\times \Gamma _n $.

\bigskip
\noindent
(iii) \textbf{For }$D_{Jd} (P\vert \vert Q)$\textbf{: }We can write
\[
D_{Jd} (P\vert \vert Q) = \frac{1}{8}J(P\vert \vert Q) - 4d(P\vert \vert Q)
 = \sum\limits_{i = 1}^n {q_i f_{Jd} \left( {\frac{p_i }{q_i }} \right)} ,
\]

\noindent where
\[
f_{Jd} (x) = \frac{1}{8}f_J \left( x \right) - 4f_d \left( x \right),
\quad
x > 0.
\]

We have
\begin{align}
\label{eq14}
{f}''_{Jd} (x) &= \frac{1}{8}{f}''_J \left( x \right) - 4{f}''_d \left( x
\right)\notag\\
& = \frac{x + 1}{8x^2} - \frac{x^{3 / 2} + 1}{2\sqrt 2 \;x^{3 / 2}(x + 1)^{3
/ 2}}\notag\\
& = \frac{1}{x^2(x + 1)\sqrt {2x + 2} }\times m_3 (x),
\end{align}

\noindent
where $m_3 (x),\;x > 0$ is given by
\[
m_3 (x) = \left( {\frac{x + 1}{2}} \right)^{5 / 2} - \sqrt x \left(
{\frac{x^{3 / 2} + 1}{2}} \right).
\]

The measures ${f}''_J \left( x \right)$ and ${f}''_d \left( x \right)$ appearing in (\ref{eq14}) are obtained from (\ref{eq10}) and (\ref{eq11}) by taking $s = 0$ (or $s = 1)$ and $s = \frac{1}{2}$ (dividing by 4) respectively. The graph of the function $m_3 (x),\;x > 0$ is given by

\begin{figure}[htbp]
\centerline{\includegraphics[width=2.00in,height=2.00in]{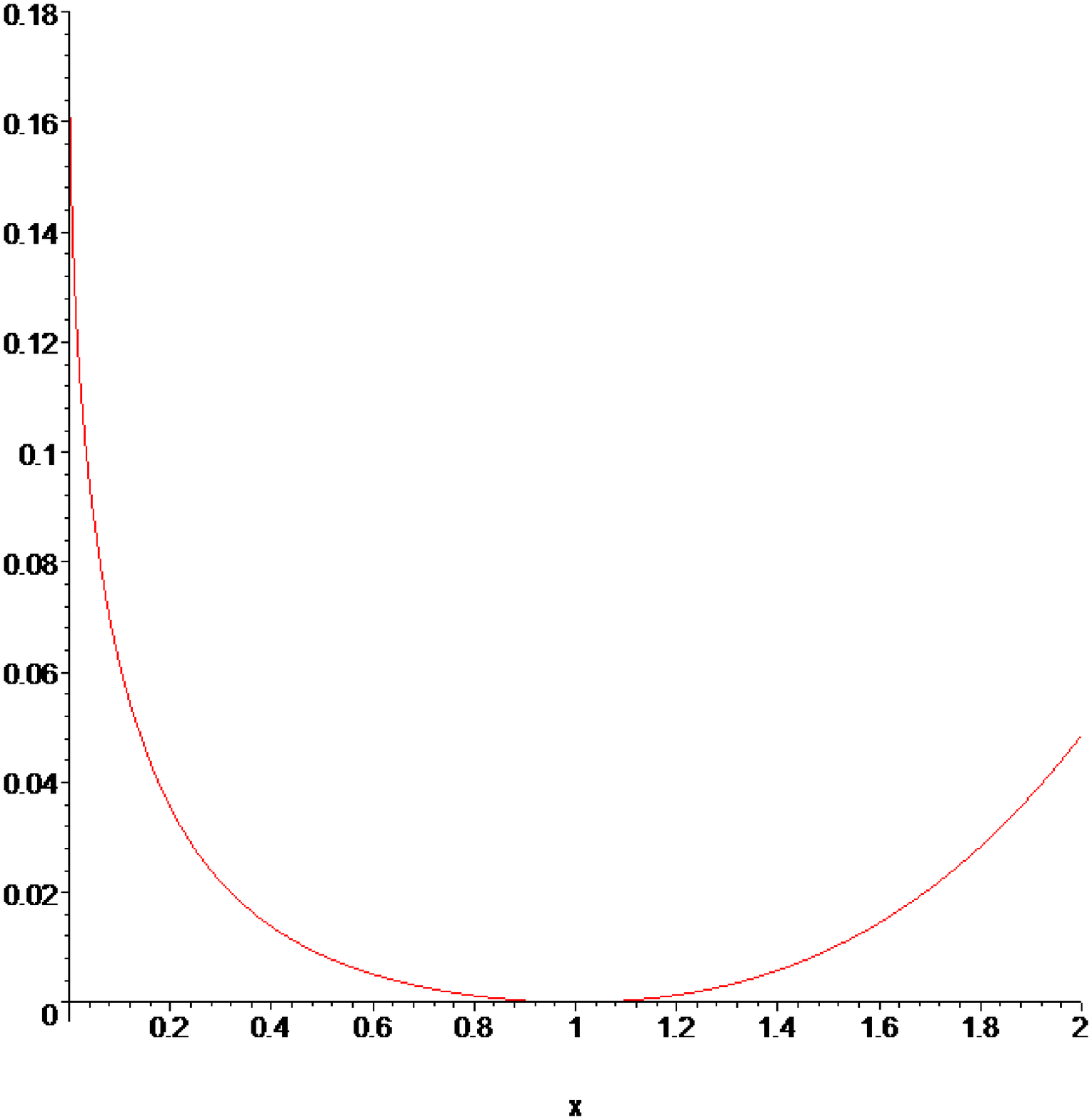}}
\label{fig3}
\end{figure}

From the above graph we observe that the function $m_3 (x) \ge 0,\;\forall x > 0$. This allows us to conclude that ${f}''_{Jd} (x) \ge 0,\;x > 0$, and hence, $f_{Jd} (x)$ is convex for all $x > 0$. Also $f_{Jd} (1) = 0$. Thus by the application of the Lemma 1.1, we conclude that measure $D_{Jd} (P\vert \vert Q)$ is nonnegative and convex for all $(P,Q) \in \Gamma _n \times \Gamma _n $.

\bigskip
\noindent
(iv) \textbf{For }$D_{dh} (P\vert \vert Q)$\textbf{: }We can write
\[
D_{dh} (P\vert \vert Q) = 4d(P\vert \vert Q) - h(P\vert \vert Q)
 = \sum\limits_{i = 1}^n {q_i f_{dh} \left( {\frac{p_i }{q_i }} \right)} ,
\]

\noindent where
\[
f_{dh} (x) = 4f_d \left( x \right) - f_h \left( x \right), \quad x > 0.
\]

We have
\begin{align}
\label{eq15}
{f}''_{dh} (x) &= 4{f}''_d \left( x \right) - {f}''_h \left( x \right)\notag\\
& = \frac{x^{3 / 2} + 1}{2\sqrt 2 \;x^{3 / 2}(x + 1)^{3 / 2}} - \frac{1}{4x^{3 / 2}}\notag\\
& = \frac{2(x^{3 / 2} + 1) - (x + 1)\sqrt {2x + 2} }{4x^{3 / 2}(x + 1)\sqrt
{2x + 2} } \notag\\
& = \frac{1}{x^{3 / 2}(x + 1)\sqrt {2x + 2} }\left[ {\frac{x^{3 / 2} + 1}{2}
- \left( {\frac{x + 1}{2}} \right)^{3 / 2}} \right], \quad \forall x > 0,
\end{align}

\noindent
where ${f}''_d \left( x \right)$ and ${f}''_h \left( x \right)$ are obtained from (\ref{eq11}) and (\ref{eq10}) by taking $s = \frac{1}{2}$(dividing by 4) and $s = \frac{1}{2}$ (dividing by 8) respectively. The non-negativity of the expression (\ref{eq15}) follows from the fact that the function $\left( {\frac{x^s + 1}{2}} \right)^{1 / s},\;s \ne 0$ is monotonically increasing function of $s$ \cite{beb}. Thus, we have ${f}''_{dh} (x) \ge 0$, $\forall x > 0$, and hence, $f_{dh} (x)$ is convex for all $x > 0$. Also $f_{dh} (1) = 0$. Thus by the application of the Lemma 1.1, we conclude that measure $D_{dh} (P\vert \vert Q)$ is nonnegative and convex for all $(P,Q) \in \Gamma _n \times \Gamma _n $.

\bigskip
\noindent
(v) \textbf{For }$D_{dI} (P\vert \vert Q)$\textbf{: }We can write
\[
D_{dI} (P\vert \vert Q) = 4d(P\vert \vert Q) - I(P\vert \vert Q)
 = \sum\limits_{i = 1}^n {q_i f_{dI} \left( {\frac{p_i }{q_i }} \right)} ,
\]

\noindent where
\[
f_{dI} (x) = f_d \left( x \right) - 4f_I \left( x \right),
\quad
x > 0.
\]

We have
\begin{align}
\label{eq16}
{f}''_{dI} (x) & = {f}''_d \left( x \right) - 4{f}''_I \left( x \right)\notag\\
& = \frac{x^{3 / 2} + 1}{2\sqrt 2 \;x^{3 / 2}(x + 1)^{3 / 2}} - \frac{1}{2x(x
+ 1)} \notag\\
& = \frac{x^{3 / 2} + 1 - \sqrt x \sqrt {2x + 2} }{2x^{3 / 2}(x + 1)\sqrt {2x
+ 2} }\notag\\
& = \frac{1}{x^{3 / 2}(x + 1)\sqrt {2x + 2} }\left[ {\left( {\frac{x^{3 / 2}
+ 1}{2}} \right) - \sqrt x \sqrt {\frac{x + 1}{2}} } \right],
\end{align}

\noindent
where ${f}''_d \left( x \right)$ and ${f}''_I \left( x \right)$ are obtained
from (\ref{eq11}) by taking $s = \frac{1}{2}$ (dividing by 4) and $s = 0$
respectively. Now we shall prove the non-negativity of the expression (\ref{eq16}).
We know that \cite{tan3}, pp.209:
\begin{center}
$\sqrt x \le \left( {\frac{\sqrt x + 1}{2}} \right)^2$ \,\, and  \,\,$\left( {\frac{\sqrt
x + 1}{2}} \right)\sqrt {\frac{x + 1}{2}} \le \left( {\frac{x + 1}{2}}
\right)$.
\end{center}

This give
\begin{equation}
\label{eq17}
\sqrt x \sqrt {\frac{x + 1}{2}} \le \left( {\frac{\sqrt x + 1}{2}}
\right)^2\sqrt {\frac{x + 1}{2}} \le \left( {\frac{x + 1}{2}} \right)\left(
{\frac{\sqrt x + 1}{2}} \right).
\end{equation}

By simple calculations, we can check that
\begin{equation}
\label{eq18}
\left( {\frac{x + 1}{2}} \right)\left( {\frac{\sqrt x + 1}{2}} \right) \le
\frac{x^{3 / 2} + 1}{2}.
\end{equation}

The expressions (\ref{eq17}) and (\ref{eq18}) together give the non-negativity of the expression (\ref{eq16}), i.e., ${f}''_{dI} (x) \ge 0$, $\forall x > 0$, and hence, $f_{dI} (x)$ is convex for all $x > 0$. Also, $f_{dI} (1) = 0$. Thus by the application of the Lemma 1.1, we conclude that measure $D_{dI}  (P\vert \vert Q)$ is nonnegative and convex for all $(P,Q) \in \Gamma _n \times \Gamma _n $.

\bigskip
\noindent
(vi) \textbf{For }$D_{d\Delta } (P\vert \vert Q)$\textbf{: }We can write
\[
D_{d\Delta } (P\vert \vert Q) = 4d(P\vert \vert Q) - \Delta (P\vert \vert Q)
 = \sum\limits_{i = 1}^n {q_i f_{d\Delta } \left( {\frac{p_i }{q_i }}
\right)} ,
\]

\noindent where
\[
f_{d\Delta } (x) = 4f_d \left( x \right) - f_\Delta \left( x \right),
\quad
x > 0.
\]

We have
\begin{align}
\label{eq19}
{f}''_{d\Delta } (x)&  = 4{f}''_d \left( x \right) - {f}''_\Delta \left( x \right)\notag\\
& = \frac{x^{3 / 2} + 1}{2\sqrt 2 \;x^{3 / 2}(x + 1)^{3 / 2}} - \frac{2}{(x + 1)^3}\notag\\
& = \frac{(x^{3 / 2} + 1)(x + 1)^2 - 4x^{3 / 2}\sqrt {2x + 2} }{2x^{3 / 2}(x + 1)^3\sqrt {2x + 2} }\notag\\
& = \frac{8}{2x^{3 / 2}(x + 1)^3\sqrt {2x + 2} }\sqrt {\frac{x + 1}{2}}
\left[ {\left( {\frac{x^{3 / 2} + 1}{2}} \right)\left( {\frac{x + 1}{2}}
\right)^{3 / 2} - x^{3 / 2}} \right]
\end{align}

\noindent
where ${f}''_d \left( x \right)$ and ${f}''_\Delta \left( x \right)$ are obtained from (\ref{eq11}) and (\ref{eq10}) by taking $s = \frac{1}{2}$ (dividing by 4) and $s = - 1$ (multiplying by 4) respectively.

\smallskip
Now we shall prove the non-negativity of the expression (\ref{eq19}). We can easily check that
\begin{equation}
\label{eq20}
\left( {\frac{\sqrt x + 1}{2}} \right)^3 \le \frac{x^{3 / 2} + 1}{2}.
\end{equation}

On the other side we know that \cite{tan3}, pp.209:
\begin{equation}
\label{eq21}
x^{3 / 2} \le \left( {\frac{\sqrt x + 1}{2}} \right)^3\left( {\frac{x +
1}{2}} \right)^{3 / 2}.
\end{equation}

Expressions (\ref{eq20}) and (\ref{eq21}) together give
\begin{equation}
\label{eq22}
x^{3 / 2} \le \left( {\frac{\sqrt x + 1}{2}} \right)^3\left( {\frac{x +
1}{2}} \right)^{3 / 2} \le \left( {\frac{x^{3 / 2} + 1}{2}} \right)\left(
{\frac{x + 1}{2}} \right)^{3 / 2}.
\end{equation}

Expression (\ref{eq22}) proves the non-negativity of the expression (\ref{eq19}), i.e., ${f}''_{d\Delta } (x) \ge 0$, $\forall x > 0$, and hence, $f_{d\Delta } (x)$ is convex for all $x > 0$. Also, $f_{d\Delta } (1) = 0$. Thus by the application of the Lemma 1.1, we conclude that measure $D_{d\Delta } (P\vert \vert Q)$ is nonnegative and convex for all $(P,Q) \in \Gamma _n \times \Gamma _n $.
\end{proof}

\section{A Sequence of Inequalities among Difference of Divergence Measures}

In this section our aim is to establish \textit{a sequence of inequalities} among \textit{difference of divergence measures}. The main result of the paper is summarized in the theorem below.

\begin{theorem} The following sequence of inequalities hold:
\begin{align}
\label{eq23}
D_{h\Delta } (P\vert \vert Q)& \le \frac{4}{5}D_{d\Delta } (P\vert \vert Q)
\le 4D_{dh} (P\vert \vert Q) \le \frac{12}{7}D_{dI} (P\vert \vert Q) \le\notag\\
 \le & 3D_{hI} (P\vert \vert Q) \le D_{Th} (P\vert \vert Q) \le
\frac{4}{3}D_{Td} (P\vert \vert Q) \le \frac{1}{4}D_{\Psi \Delta } (P\vert
\vert Q) \le \notag\\
& \le \frac{1}{3}D_{\Psi h} (P\vert \vert Q) \le \frac{4}{11}D_{\Psi d}
(P\vert \vert Q) \le \frac{1}{2}D_{\Psi T} (P\vert \vert Q).
\end{align}
\end{theorem}

The proof of the above theorem is based on the following propositions. In order to prove the propositions, we shall be using frequently, expressions (\ref{eq10}) and (\ref{eq11}) with particular values.

\begin{proposition} We have
\[
D_{h\Delta } (P\vert \vert Q) \le \frac{4}{5}D_{d\Delta } (P\vert \vert Q).
\]
\end{proposition}

\begin{proof} Let us consider
\begin{align}
g_{h\Delta \_d\Delta } (x) &= \frac{{f}''_{h\Delta } (x)}{{f}''_{d\Delta }
(x)} = \frac{\sqrt {(2x + 2)} \left[ {(x + 1)^3 - 8x^{3 / 2}}
\right]}{2\left[ {(x^{3 / 2} + 1)(x + 1)^2 - 4x^{3 / 2}\sqrt {2x + 2} }
\right]} \notag\\
& = \frac{\sqrt {(2x + 2)} \left[ {\left( {\sqrt x + 1} \right)^2(x + 1) +
4x\left( {\sqrt x - 1} \right)^2} \right]}{2\left[ {(x^{3 / 2} + 1)(x + 1)^2
- 4x^{3 / 2}\sqrt {2x + 2} } \right]},  \quad x \ne 1\notag
\end{align}

\noindent
for all $x \in (0,\infty )$, where ${f}''_{h\Delta } (x) = {\varphi }''_{ -
1} (x) - {\psi }''_{ - 1} (x)$ and ${f}''_{d\Delta } (x) = {\psi }''_{1 / 2}
(x) - {\psi }''_{ - 1} (x)$. Calculating the first order derivative of the
function $g_{I\Delta \_h\Delta } (x)$ with respect to $x$, one gets
\begin{equation}
\label{eq24}
{g}'_{h\Delta \_d\Delta } (x) = \frac{ - 3(x + 1)\left( {\sqrt x - 1}
\right)\sqrt {2x + 2} }{4\left[ {(x^{3 / 2} + 1)(x + 1)^2 - 4x^{3 / 2}\sqrt
{2x + 2} } \right]^2}\times k_1 (x),
\end{equation}

\noindent where
\begin{align}
\label{eq25}
k_1 (x) = x^3 - 8x^{5 / 2} - &5x^2 - 8x^{3 / 2} - 5x - 8x^{1 / 2}  + 1 +\notag\\
& + \,4\sqrt {2x(x + 1)} \left( {\sqrt x + 1} \right)\left( {x + 1} \right).
\end{align}

The graph of the function $k_1 (x)$, $x > 0$ is given by

\begin{figure}[htbp]
\centerline{\includegraphics[width=2.00in,height=2.00in]{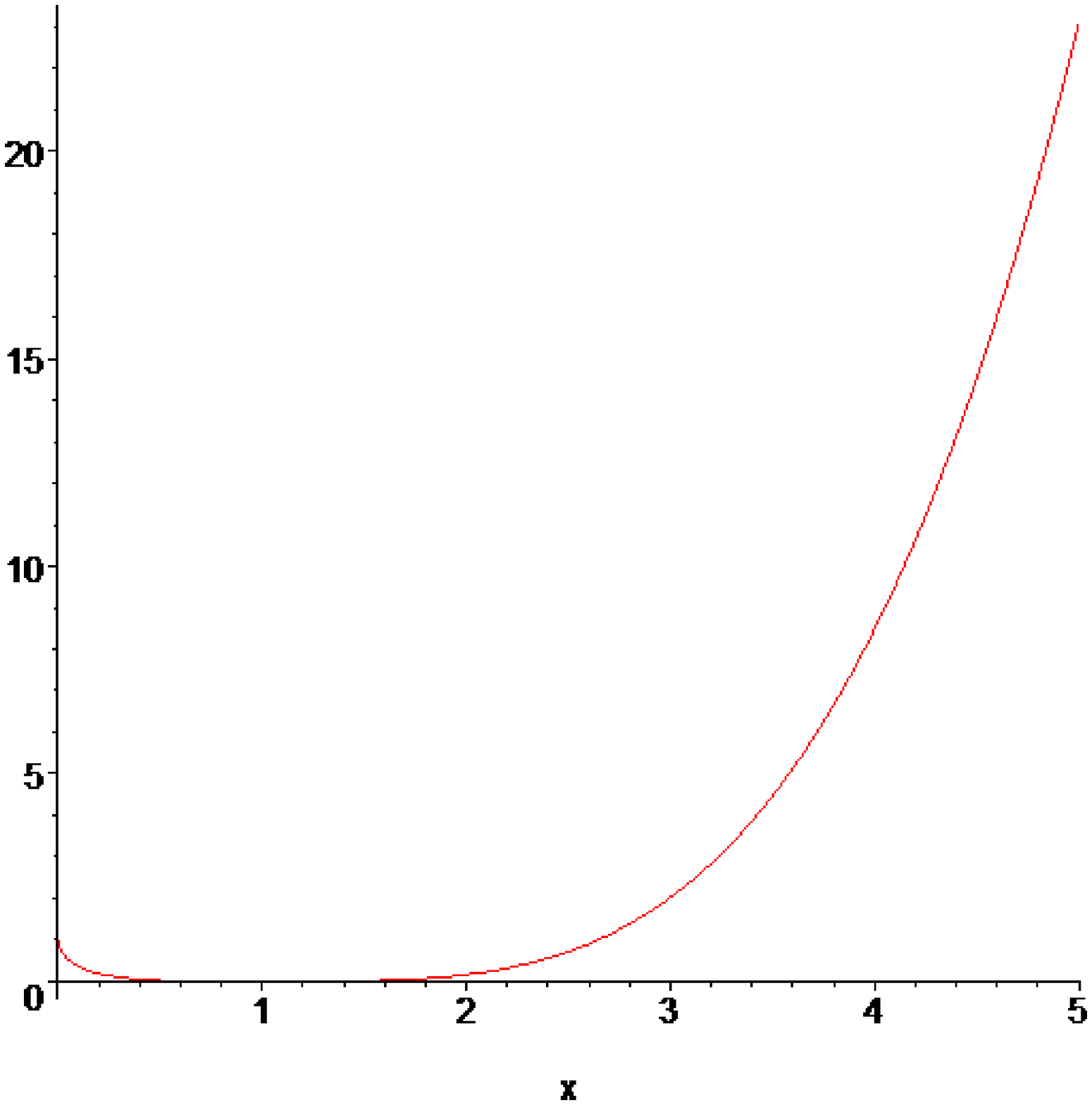}}
\label{fig4}
\end{figure}

We observe from the above graph that the function $k_1 (x) \ge 0$, $\forall x > 0$. Thus from the from the expression (\ref{eq24}), we conclude that
\begin{equation}
\label{eq26}
{g}'_{h\Delta \_d\Delta } (x) = \begin{cases}
 { > 0,} & {x < 1} \\
 { < 0,} & {x > 1} \\
\end{cases}
\end{equation}

Let us calculate now $g_{h\Delta \_d\Delta } (1)$. We observe that
\[
\left. {g_{h\Delta \_d\Delta } (x)} \right|_{x = 1} = \left.
{\frac{{f}''_{h\Delta } (x)}{{f}''_{d\Delta } (x)}} \right|_{x = 1} = \left.
{\frac{\left( {{f}''_{h\Delta } (x)} \right)^\prime }{\left( {{f}''_{d\Delta
} (x)} \right)^\prime }} \right|_{x = 1} = \mbox{indermination}.
\]

Calculating the second order derivatives of numerator and denominator of the function $g_{I\Delta \_h\Delta } (x)$, we have
\begin{equation}
\label{eq27}
g_{h\Delta \_d\Delta } (1) = \left. {\frac{\left( {{f}''_{h\Delta } (x)}
\right)^{\prime \prime }}{\left( {{f}''_{d\Delta } (x)} \right)^{\prime
\prime }}} \right|_{x = 1} = \frac{\textstyle{3 \over 4}}{\textstyle{{15}
\over {16}}} = \frac{4}{5}.
\end{equation}

By the application of the inequalities (\ref{eq5}) with (\ref{eq27}) we get the required result.
\end{proof}

\begin{proposition} We have
\[
D_{d\Delta } (P\vert \vert Q) \le 5\,D_{dh} (P\vert \vert Q).
\]
\end{proposition}

\begin{proof} Let us consider
\[
g_{d\Delta \_dh} (x) = \frac{{f}''_{d\Delta } (x)}{{f}''_{dh} (x)} =
\frac{2\left[ {(x + 1)^2(x^{3 / 2} + 1) - 4x^{3 / 2}\sqrt {2x + 2} }
\right]}{(x + 1)^2\left[ {2(x^{3 / 2} + 1) - (x + 1)\sqrt {2x + 2} }
\right]} , \quad x \ne 1,
\]

\noindent
for all $x \in (0,\infty )$, where ${f}''_{d\Delta } (x) = {\psi }''_{1 / 2}
(x) - {\psi }''_{ - 1} (x)$ and ${f}''_{dh} (x) = {\psi }''_{1 / 2} (x) -
{\varphi }''_{1 / 2} (x)$.

\smallskip
Calculating the first order derivative of the function $g_{d\Delta \_dh} (x)$ with respect to $x$, one gets
\begin{equation}
\label{eq28}
{g}'_{d\Delta \_dh} (x) = \frac{ - 3\left( {\sqrt x - 1} \right)\sqrt {2x +
2} }{(x + 1)^3\left[ {2(x^{3 / 2} + 1) - (x + 1)\sqrt {2x + 2} }
\right]^2}\times k_1 (x),
\end{equation}

\noindent
where $k_1 (x)$, $x > 0$ is as given in (\ref{eq25}). Since $k_1 (x) \ge 0$,
$\forall x > 0$. Thus from the from the expression (\ref{eq28}), we conclude that
\begin{equation}
\label{eq29}
{g}'_{d\Delta \_dh} (x) = \begin{cases}
 { > 0,} & {x < 1} \\
 { < 0,} & {x > 1} \\
\end{cases}
\end{equation}

Let us calculate now $g_{d\Delta \_dh} (1)$. We observe that
\[
\left. {g_{d\Delta \_dh} (x)} \right|_{x = 1} = \left. {\frac{{f}''_{d\Delta
} (x)}{{f}''_{dh} (x)}} \right|_{x = 1} = \left. {\frac{\left(
{{f}''_{d\Delta } (x)} \right)^\prime }{\left( {{f}''_{dh} (x)}
\right)^\prime }} \right|_{x = 1} = \mbox{indermination}.
\]

Calculating the second order derivatives of numerator and denominator of the function $g_{d\Delta \_dh} (x)$, we have
\begin{equation}
\label{eq30}
g_{d\Delta \_dh} (1) = \left. {\frac{\left( {{f}''_{d\Delta } (x)}
\right)^{\prime \prime }}{\left( {{f}''_{dh} (x)} \right)^{\prime \prime }}}
\right|_{x = 1} = \frac{\textstyle{{15} \over {16}}}{\textstyle{3 \over
{16}}} = 5.
\end{equation}

By the application of the inequalities (\ref{eq5}) with (\ref{eq30}) we get the requires result.
\end{proof}

\begin{proposition} We have
\[
D_{dh} (P\vert \vert Q) \le \frac{3}{7}D_{dI} (P\vert \vert Q).
\]
\end{proposition}

\begin{proof} Let us consider
\[
g_{dh\_dI} (x) = \frac{{f}''_{dh} (x)}{{f}''_{dI} (x)} = \frac{2(x^{3 / 2} +
1) - (x + 1)\sqrt {2x + 2} }{2\left[ {(x^{3 / 2} + 1) - \sqrt {2x(x + 1)} }
\right]}  , \quad x \ne 1
\]

\noindent
for all $x \in (0,\infty )$, where ${f}''_{dh} (x) = {\psi }''_{1 / 2} (x) -
{\varphi }''_{1 / 2} (x)$ and ${f}''_{dI} (x) = {\psi }''_{1 / 2} (x) -
{\varphi }''_1 (x)$. Calculating the first order derivative of the function
$g_{dh\_dI} (x)$ with respect to $x$, one gets
\begin{equation}
\label{eq31}
{g}'_{dh\_dI} (x) = - \frac{\left( {\sqrt x - 1} \right)\sqrt {2x + 2}
}{4\sqrt x (x + 1)\left[ {\left( {x^{3 / 2} + 1) - \sqrt {2x(x + 1)} }
\right)} \right]^2}\times k_2 (x)
\end{equation}

\noindent where
\begin{equation}
\label{eq32}
k_2 (x) = x^2 - x^{3 / 2} + 6x - \sqrt x + 2 - \left( {x + 1} \right)\left(
{\sqrt x + 1} \right)\sqrt {2x + 2} .
\end{equation}

The graph of the function $k_2 (x)$, $x > 0$ is given by
\newpage

\begin{figure}[htbp]
\centerline{\includegraphics[width=2.00in,height=2.00in]{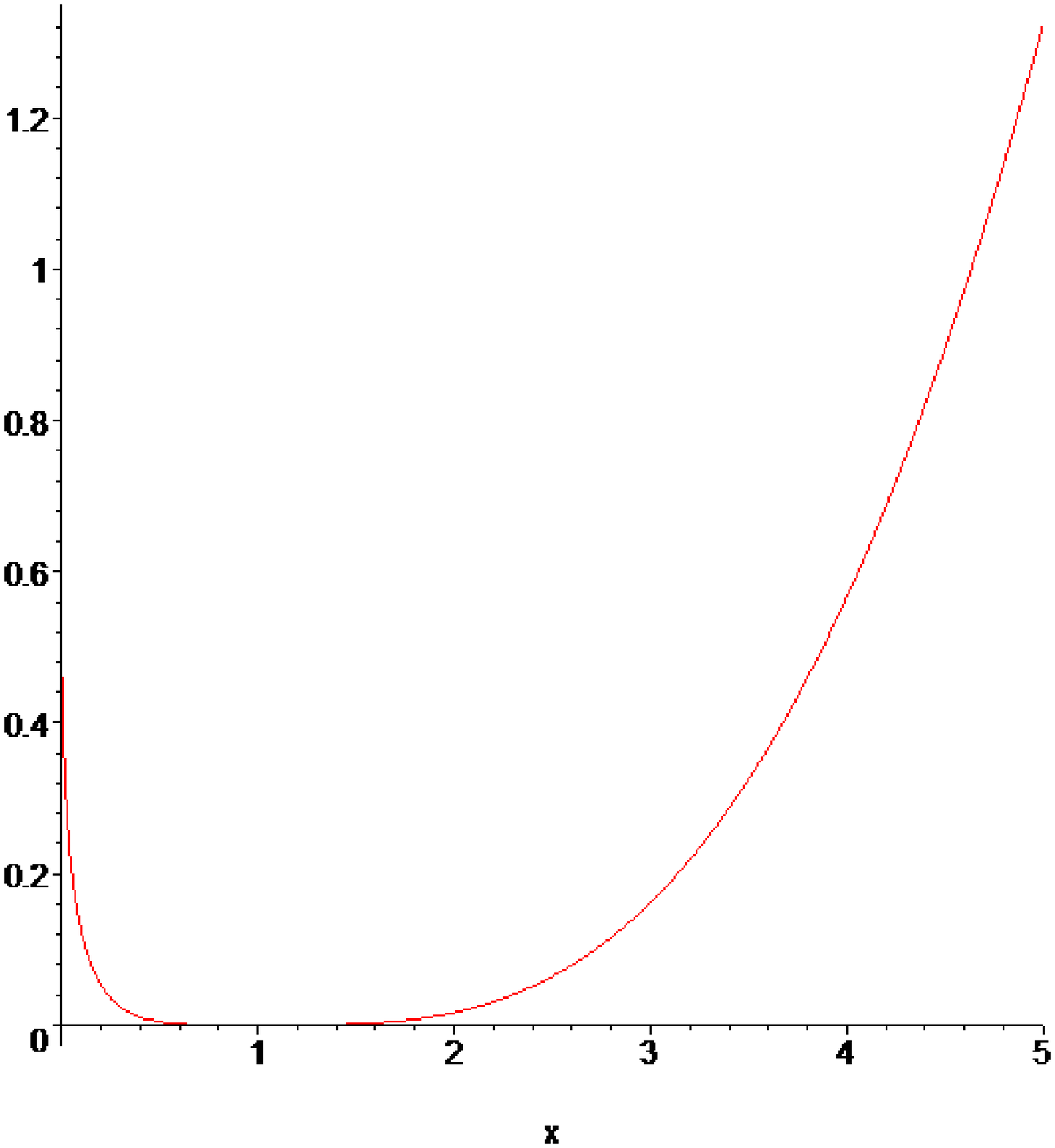}}
\label{fig5}
\end{figure}

We observe from the above graph that the function $k_2 (x) \ge 0$, $\forall
x > 0$. Thus from the from the expression (\ref{eq37}), we conclude that
\begin{equation}
\label{eq33}
{g}'_{dh\_dI} (x) = \begin{cases}
 { > 0,} & {x < 1} \\
 { < 0,} & {x > 1} \\
\end{cases}
\end{equation}

Let us calculate now $g_{dh\_dI} (1)$. We observe that
\[
\left. {g_{dh\_dI} (x)} \right|_{x = 1} = \left. {\frac{{f}''_{dh}
(x)}{{f}''_{dI} (x)}} \right|_{x = 1} = \left. {\frac{\left( {{f}''_{dh}
(x)} \right)^\prime }{\left( {{f}''_{dI} (x)} \right)^\prime }} \right|_{x =
1} = \mbox{indermination}.
\]

Calculating the second order derivatives of numerator and denominator of the
function $g_{dh\_dI} (x)$, we have
\begin{equation}
\label{eq34}
g_{dh\_dI} (1) = \left. {\frac{\left( {{f}''_{dh} (x)} \right)^{\prime
\prime }}{\left( {{f}''_{dI} (x)} \right)^{\prime \prime }}} \right|_{x = 1}
= \frac{\textstyle{3 \over {64}}}{\textstyle{7 \over {64}}} = \frac{3}{7}.
\end{equation}

By the application of the inequalities (\ref{eq5}) with (\ref{eq34}) we get the requires result.
\end{proof}

\begin{proposition} We have
\[
D_{dI} (P\vert \vert Q) \le \frac{7}{4}D_{hI} (P\vert \vert Q).
\]
\end{proposition}

\begin{proof} Let us consider
\[
g_{dI\_hI} (x) = \frac{{f}''_{dI} (x)}{{f}''_{hI} (x)} = \frac{2\left( {x^{3
/ 2} + 1 - \sqrt {2x(x + 1)} } \right)}{\sqrt {2x + 2} \left( {\sqrt x - 1}
\right)^2} , \quad  x \ne 1,
\]

\noindent
for all $x \in (0,\infty )$, where ${f}''_{dI} (x) = {\psi }''_{1 / 2} (x) -
{\varphi }''_1 (x)$ and ${f}''_{hI} (x) = {\varphi }''_{1 / 2} (x) -
{\varphi }''_1 (x)$. Calculating the first order derivative of the function
$g_{dI\_hd} (x)$ with respect to $x$, one gets
\begin{equation}
\label{eq35}
{g}'_{dI\_hI} (x) = - \frac{\left( {\sqrt x - 1} \right)\sqrt x }{x(x +
1)\sqrt {2x + 2} \left( {2\sqrt x - x - 1} \right)^2}\times k_2 (x),  \quad x \ne 1.
\end{equation}

\noindent
where $k_2 (x)$, $x > 0$ is as given by (\ref{eq38}). Since $k_2 (x) \ge 0$,
$\forall x > 0$. Thus from the from the expression (\ref{eq35}), we conclude that
\begin{equation}
\label{eq36}
{g}'_{dI\_hI} (x) = \begin{cases}
 { > 0,} & {x < 1} \\
 { < 0,} & {x > 1} \\
\end{cases}
\end{equation}

Let us calculate now $g_{dI\_hI} (1)$. We observe that
\[
\left. {g_{dI\_hI} (x)} \right|_{x = 1} = \left. {\frac{{f}''_{dI}
(x)}{{f}''_{dI} (x)}} \right|_{x = 1} = \left. {\frac{\left( {{f}''_{dI}
(x)} \right)^\prime }{\left( {{f}''_{hI} (x)} \right)^\prime }} \right|_{x =
1} = \mbox{indermination}.
\]

Calculating the second order derivatives of numerator and denominator of the
function $g_{dI\_hI} (x)$, we have
\begin{equation}
\label{eq37}
g_{dI\_hI} (1) = \left. {\frac{\left( {{f}''_{dI} (x)} \right)^{\prime
\prime }}{\left( {{f}''_{hI} (x)} \right)^{\prime \prime }}} \right|_{x = 1}
= \frac{\textstyle{7 \over {64}}}{\textstyle{1 \over {16}}} = \frac{7}{4}.
\end{equation}

By the application of the inequalities (\ref{eq5}) with (\ref{eq37}) we get the required result.
\end{proof}

\begin{proposition}  We have
\[
D_{Th} (P\vert \vert Q) \le \frac{4}{3}D_{Td} (P\vert \vert Q).
\]
\end{proposition}

\begin{proof} Let us consider
\[
g_{Th\_Td} (x) = \frac{{f}''_{Th} (x)}{{f}''_{Td} (x)} = \frac{\left( {\sqrt
x - 1} \right)^2\left( {x + \sqrt x + 1} \right)\sqrt {2x + 2} }{(x^2 +
1)\sqrt {2x + 2} - 2\sqrt x \left( {x^{3 / 2} + 1} \right)} , \quad  x \ne 1
\]

\noindent
for all $x \in (0,\infty )$, where ${f}''_{Th} (x) = {\psi }''_0 (x) -
{\varphi }''_{1 / 2} (x)$ and ${f}''_{Td} (x) = {\psi }''_0 (x) - {\psi
}''_{1 / 2} (x)$. Calculating the first order derivative of the function
$g_{Th\_Td} (x)$ with respect to $x$, one gets
\begin{equation}
\label{eq38}
{g}'_{Th\_Td} (x) = \frac{ - 2x^2(x + 1)\left( {\sqrt x - 1} \right)\sqrt
{2x + 2} }{2\sqrt x \left[ {2\sqrt x \left( {x^{3 / 2} + 1} \right) - (x^2 +
1)\sqrt {x + 1} } \right]^2}\times k_3 (x)
\end{equation}

\noindent where
\begin{align}
k_3 (x) = 2(x^2 + 1)(x^2 + x + 1) &+ 2\sqrt x (x + 1)(x^2 + 7x + 1)\notag\\
& - \,(x + 1)(x^2 + 4x + 1)\left( {\sqrt x + 1} \right)\sqrt {2x + 2} .\notag
\end{align}

The graph of the function $k_3 (x), x > 0$ is given by
\newpage

\begin{figure}[htbp]
\centerline{\includegraphics[width=2.00in,height=2.00in]{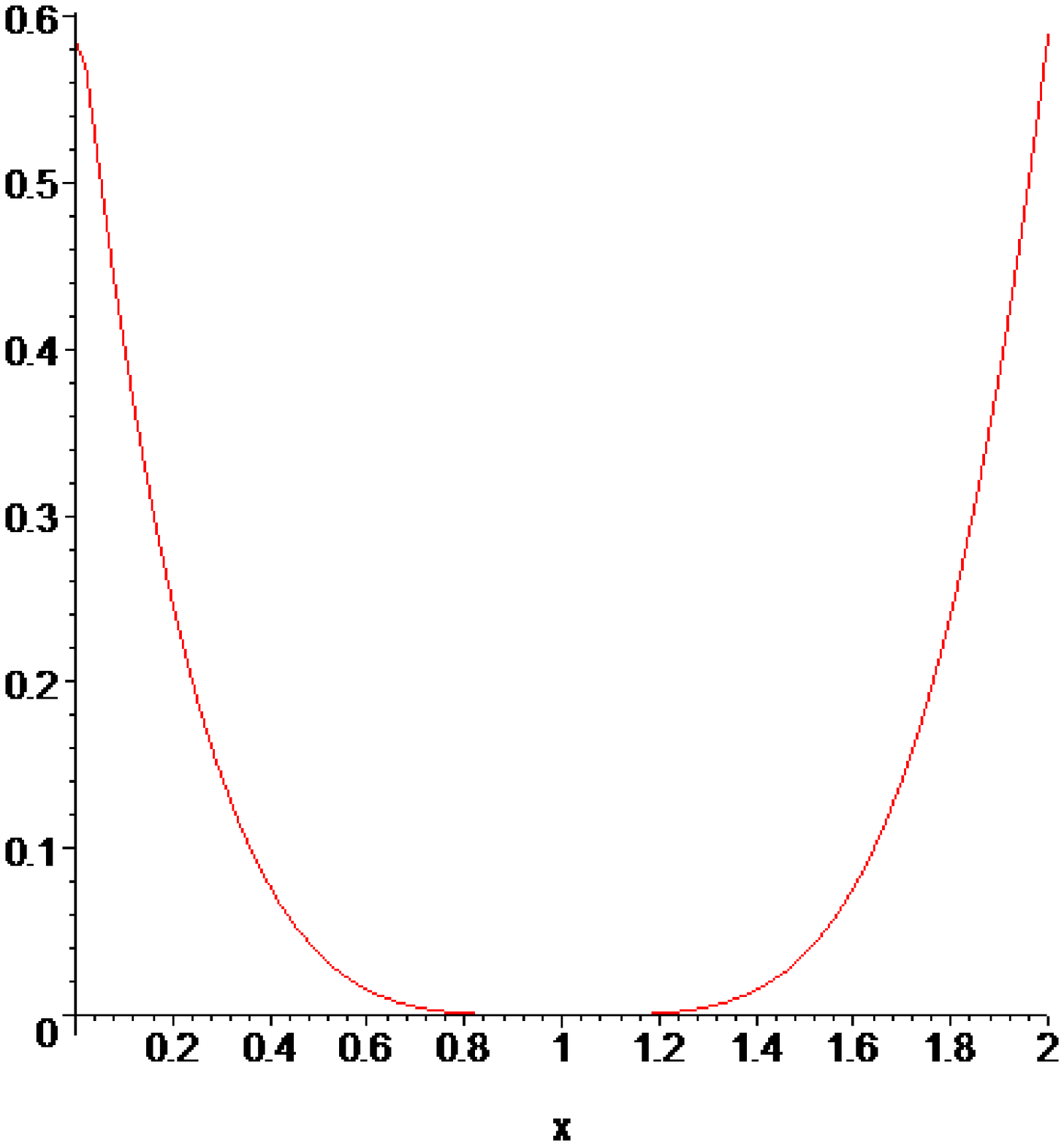}}
\label{fig6}
\end{figure}

We observe from the above graph that the function $k_3 (x) \ge 0$, $\forall
x > 0$. Thus from the from the expression (\ref{eq38}), we conclude that
\begin{equation}
\label{eq39}
{g}'_{Th\_Td} (x) = \begin{cases}
 { > 0,} & {x < 1} \\
 { < 0,} & {x > 1} \\
\end{cases}
\end{equation}

Let us calculate now $g_{Th\_Td} (1)$. We observe that
\[
\left. {g_{Th\_Td} (x)} \right|_{x = 1} = \left. {\frac{{f}''_{Th}
(x)}{{f}''_{Td} (x)}} \right|_{x = 1} = \left. {\frac{\left( {{f}''_{Th}
(x)} \right)^\prime }{\left( {{f}''_{Td} (x)} \right)^\prime }} \right|_{x =
1} = \mbox{indermination}.
\]

Calculating the second order derivatives of numerator and denominator of the function $g_{Th\_Td} (x)$, we have
\begin{equation}
\label{eq40}
g_{Th\_Td} (1) = \left. {\frac{\left( {{f}''_{Th} (x)} \right)^{\prime
\prime }}{\left( {{f}''_{Td} (x)} \right)^{\prime \prime }}} \right|_{x = 1}
= \frac{\textstyle{3 \over {16}}}{\textstyle{9 \over {64}}} = \frac{4}{3}.
\end{equation}

By the application of the inequalities (\ref{eq5}) with (\ref{eq40}) we get the required result.
\end{proof}

\begin{proposition} We have
\[
D_{Td} (P\vert \vert Q) \le \frac{3}{16}D_{\Psi \Delta } (P\vert \vert Q).
\]
\end{proposition}

\begin{proof} Let us consider
\begin{align}
g_{Td\_\Psi \Delta } (x) = \frac{{f}''_{Td} (x)}{{f}''_{\Psi \Delta } (x)} & =
\frac{2x(x + 1)^2\left[ {(x^2 + 1)(2x + 2)^{3 / 2} - 4\sqrt x (x + 1)(x^{3 /
2} + 1)} \right]}{(x - 1)^2(2x + 2)^{3 / 2}\left( {x^4 + 5x^3 + 12x^2 + 5x +
1} \right)}\notag\\
& = \frac{2x(x + 1)^2\left[ {(x^2 + 1)\sqrt {2x + 2} - 2\sqrt x \left( {x^{3
/ 2} + 1} \right)} \right]}{(x - 1)^2\sqrt {2x + 2} \left( {x^4 + 5x^3 +
12x^2 + 5x + 1} \right)}, \quad x \ne 1 \notag
\end{align}

\noindent
for all $x \in (0,\infty )$, where ${f}''_{Td} (x) = {\psi }''_0 (x) - {\psi }''_{1 / 2} (x)$ and ${f}''_{\Psi \Delta } (x) = {\psi }''_2 (x) - {\psi }''_{ - 1} (x)$. Calculating the first order derivative of the function $g_{Td\_\Psi \Delta } (x)$ with respect to $x$, one gets
\begin{equation}
\label{eq41}
{g}'_{Td\_\Psi \Delta } (x) = - \frac{2}{x(x - 1)^3(2x + 2)^{5 / 2}(x +
1)^3\left( {x^4 + 5x^3 + 12x^2 + 5x + 1} \right)^2}\times k_4 (x),
\end{equation}

\noindent
where $k_4 (x)$ is given by
\begin{align}
k_4 (x) = - 12 &\sqrt x \left( {\sqrt x + 1} \right)\left( {x + 1}
\right)^2 \left[ {(x + 1)(x^6 + 4x^5 + 6x^4 + 18x^3 + 6x^2 + 4x + 1)}
\right.\notag\\
& \left. { - x^{(1 / 2)}(x^2 + 1)(x^4 + 3x^3 + 3x + 1)} \right] +\notag\\
& + \,(2x + 2)^{(5 / 2)}\left( {1 + 4x + 10x^2 + 52x^3 + 58x^4 + 52x^5 +
10x^6 + 4x^7 + x^8} \right).\notag
\end{align}

The graph of the function $k_4 (x)$, $x > 0$ is given by

\begin{figure}[htbp]
\centerline{\includegraphics[width=2.00in,height=2.00in]{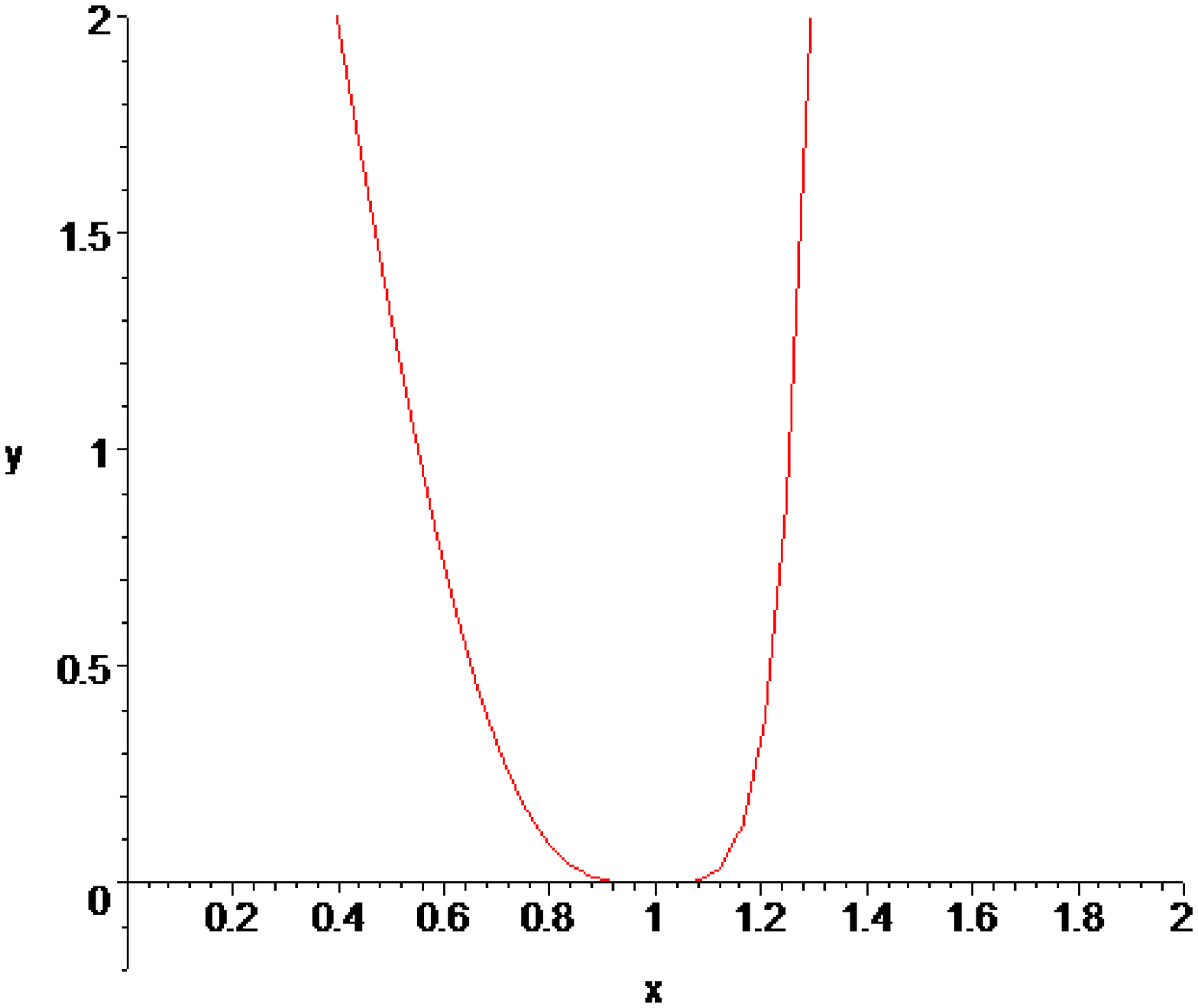}}
\label{fig7}
\end{figure}

We observe from the above graph that the function $k_4 (x) \ge 0$, $\forall
x > 0$. Thus from the from the expression (\ref{eq41}), we conclude that
\begin{equation}
\label{eq42}
{g}'_{Td\_\Psi \Delta } (x) = \begin{cases}
 { > 0,} & {x < 1} \\
 { < 0,} & {x > 1} \\
\end{cases}
\end{equation}

Let us calculate now $g_{Td\_\Psi \Delta } (1)$. We observe that
\[
\left. {g_{Td\_\Psi \Delta } (x)} \right|_{x = 1} = \left. {\frac{{f}''_{Td}
(x)}{{f}''_{\Psi \Delta } (x)}} \right|_{x = 1} = \left. {\frac{\left(
{{f}''_{Td} (x)} \right)^\prime }{\left( {{f}''_{\Psi \Delta } (x)}
\right)^\prime }} \right|_{x = 1} = \mbox{indermination}.
\]

Calculating the second order derivatives of numerator and denominator of the
function $g_{Td\_\Psi \Delta } (x)$, we have
\begin{equation}
\label{eq43}
g_{Td} (1) = \left. {\frac{\left( {{f}''_{Td} (x)} \right)^{\prime \prime
}}{\left( {{f}''_{\Psi \Delta } (x)} \right)^{\prime \prime }}} \right|_{x =
1} = \frac{\textstyle{9 \over 4}}{12} = \frac{3}{16}.
\end{equation}

By the application of the inequalities (\ref{eq5}) with (\ref{eq43}) we get the required result.
\end{proof}

\begin{proposition} We have
\[
D_{\Psi h} (P\vert \vert Q) \le \frac{12}{11}D_{\Psi d} (P\vert \vert Q).
\]
\end{proposition}

\begin{proof} Let us consider
\[
g_{\Psi h\_\Psi d} (x) = \frac{{f}''_{\Psi h} (x)}{{f}''_{\Psi d} (x)} =
\frac{\left( {x^{3 / 2} - 1} \right)^2(2x + 2)^{3 / 2}}{(x^3 + 1)(2x + 2)^{3
/ 2} - 8x^{3 / 2}(x^{3 / 2} + 1)} , \quad x \ne 1
\]

\noindent
for all $x \in (0,\infty )$, where ${f}''_{\Psi h} (x) = {\psi }''_2 (x) -
{\varphi }''_{1 / 2} (x)$ and ${f}''_{\Psi d} (x) = {\psi }''_2 (x) - {\psi
}''_{1 / 2} (x)$. Calculating the first order derivative of the function
$g_{\Psi h\_\Psi d} (x)$ with respect to $x$, one gets
\begin{equation}
\label{eq44}
{g}'_{\Psi h\_\Psi d} (x) = - \frac{3\left( {\sqrt x - 1} \right)\left( {x +
\sqrt x + 1} \right)\sqrt {2x + 2} }{\sqrt x \left[ {(x^3 + 1)(2x + 2)^{3 /
2} - 8x^3(x^{3 / 2} + 1)} \right]^2}\times k_5 (x)
\end{equation}

\noindent where
\[
k_5 (x) = \left[ {8\left( {x^4 + 3x^{5 / 2} + 3x^{3 / 2} + 1} \right) -
\left( {x^{3 / 2} + 1} \right)\left( {2x + 2} \right)^{5 / 2}} \right].
\]

The graph of the function $k_5 (x)$, $x > 0$ is given by

\begin{figure}[htbp]
\centerline{\includegraphics[width=2.00in,height=2.00in]{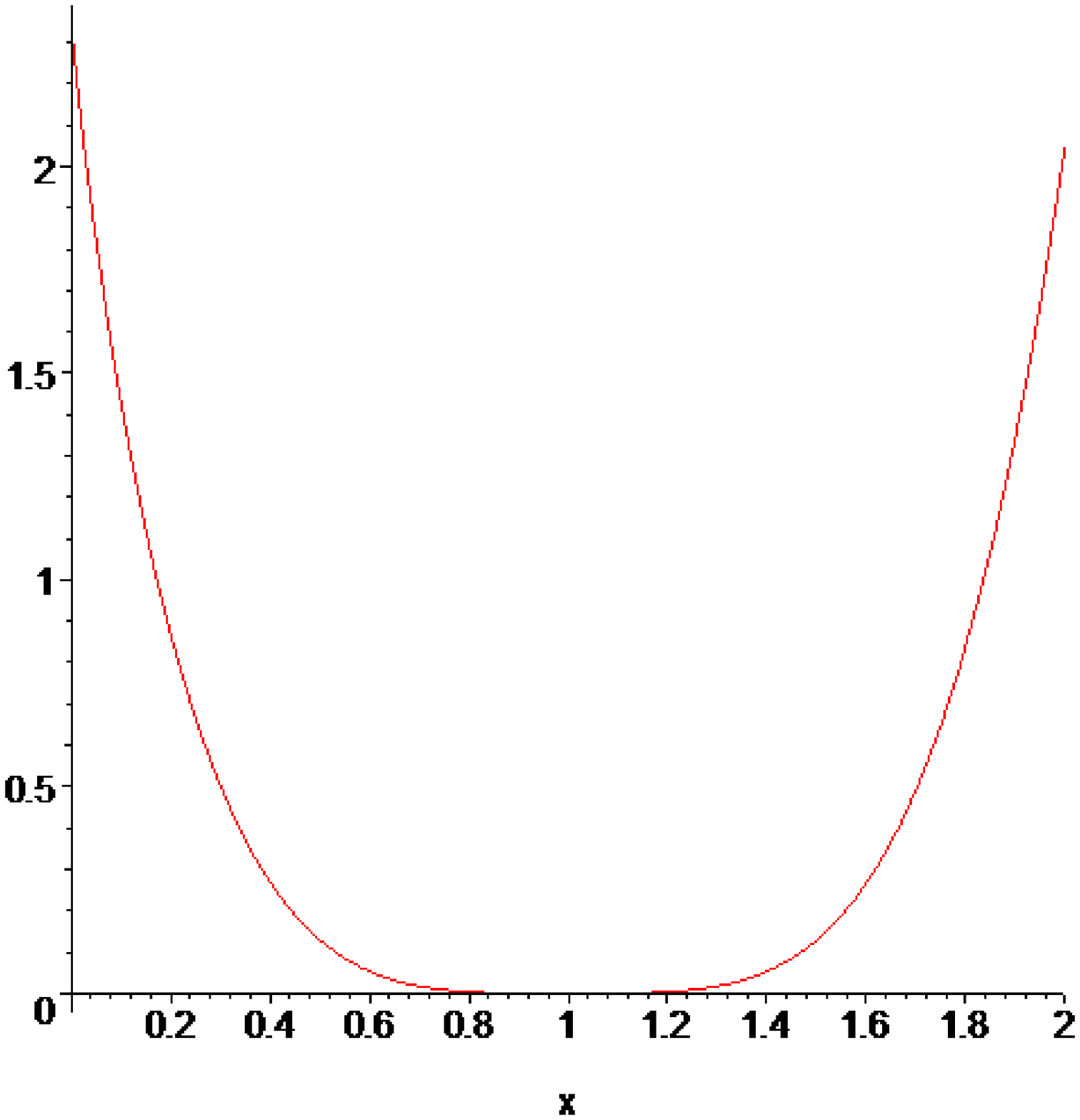}}
\label{fig8}
\end{figure}

We observe from the above graph that the function $k_5 (x) \ge 0$, $\forall
x > 0$. Thus from the from the expression (\ref{eq44}), we conclude that
\begin{equation}
\label{eq45}
{g}'_{\Psi h\_\Psi d} (x) = \begin{cases}
 { > 0,} & {x < 1} \\
 { < 0,} & {x > 1} \\
\end{cases}
\end{equation}

Let us calculate now $g_{\Psi h\_\Psi d} (1)$. We observe that
\[
\left. {g_{\Psi h\_\Psi d\Delta } (x)} \right|_{x = 1} = \left.
{\frac{{f}''_{\Psi h} (x)}{{f}''_{\Psi d} (x)}} \right|_{x = 1} = \left.
{\frac{\left( {{f}''_{\Psi h} (x)} \right)^\prime }{\left( {{f}''_{\Psi d}
(x)} \right)^\prime }} \right|_{x = 1} = \mbox{indermination}.
\]

Calculating the second order derivatives of numerator and denominator of the
function $g_{\Psi h\_\Psi d} (x)$, we have
\begin{equation}
\label{eq46}
g_{\Psi h\_\Psi d} (1) = \left. {\frac{\left( {{f}''_{\Psi h} (x)}
\right)^{\prime \prime }}{\left( {{f}''_{\Psi d} (x)} \right)^{\prime \prime
}}} \right|_{x = 1} = \frac{9}{\textstyle{{33} \over 4}} = \frac{12}{11}.
\end{equation}

By the application of the inequalities (\ref{eq5}) with (\ref{eq46}) we get the required result.
\end{proof}

\begin{proposition} We have
\[
D_{\Psi d} (P\vert \vert Q) \le \frac{11}{8}D_{\Psi T} (P\vert \vert Q).
\]
\end{proposition}

\begin{proof} Let us consider
\[
g_{\Psi d\_\Psi T} (x) = \frac{{f}''_{\Psi d} (x)}{{f}''_{\Psi T} (x)} =
\frac{(x + 1)\left[ {(x^3 + 1)(2x + 2)^{3 / 2} - 8x^{3 / 2}(x^{3 / 2} + 1)}
\right]}{(2x + 2)^{3 / 2}(x - 1)^2(x^2 + x + 1)} , \quad x \ne 1
\]

\noindent
for all $x \in (0,\infty )$, where ${f}''_{\Psi d} (x) = {\psi }''_2 (x) -
{\psi }''_{1 / 2} (x)$ and ${f}''_{\Psi T} (x) = {\psi }''_2 (x) - {\psi
}''_0 (x)$. Calculating the first order derivative of the function $g_{\Psi
d\_\Psi T} (x)$ with respect to $x$, one gets
\begin{equation}
\label{eq47}
{g}'_{\Psi d\_\Psi T} (x) = - \frac{2x^{5 / 2}\left( {\sqrt x - 1}
\right)\left( {\sqrt x + 1} \right)\left( {x + 1} \right)}{x^{5 / 2}(x -
1)^4(x^2 + x + 1)^2(2x + 2)^{5 / 2}}\times k_6 (x),
\end{equation}

\noindent
where $k_6 (x)$, $x > 0$ is given by
\begin{align}
k_6 (x) = (2x + 2)^{5 / 2}(x^4 + 4x^2 + 1) &- 4x^2(3x^4 + 4x^3 + 4x^2 + 7x + 6)\notag\\
& - \,4\sqrt x (3 + 4x + 4x^2 + 7x^3 + 6x^4).\notag
\end{align}

The graph of $k_6 (x)$, $x > 0$ is given by
\begin{figure}[htbp]
\centerline{\includegraphics[width=2.00in,height=2.00in]{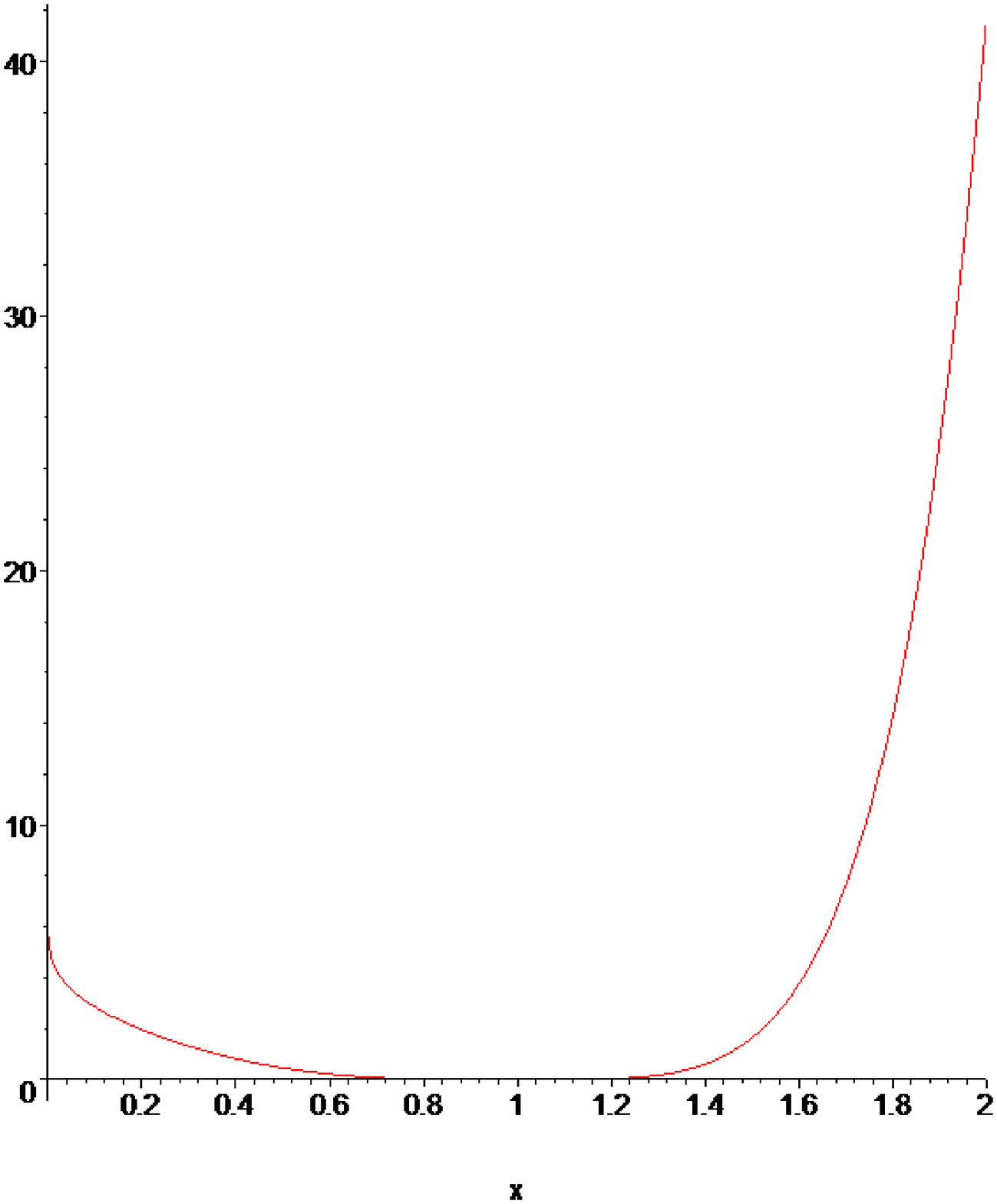}}
\label{fig9}
\end{figure}

We observe from the above graph that the function $k_6 (x) \ge 0$, $\forall x > 0$. Thus from the from the expression (\ref{eq47}), we conclude that
\begin{equation}
\label{eq48}
{g}'_{\Psi d\_\Psi T} (x) = \begin{cases}
 { > 0,} & {x < 1} \\
 { < 0,} & {x > 1} \\
\end{cases}
\end{equation}

Let us calculate now $g_{\Psi d\_\Psi T} (1)$. We observe that
\[
\left. {g_{\Psi d\_\Psi T} (x)} \right|_{x = 1} = \left. {\frac{{f}''_{\Psi
d} (x)}{{f}''_{\Psi T} (x)}} \right|_{x = 1} = \left. {\frac{\left(
{{f}''_{\Psi d} (x)} \right)^\prime }{\left( {{f}''_{\Psi T} (x)}
\right)^\prime }} \right|_{x = 1} = \mbox{indermination}.
\]

Calculating the second order derivatives of numerator and denominator of the function $g_{\Psi d\_\Psi T} (x)$, we have
\begin{equation}
\label{eq49}
g_{\Psi d\_\Psi T} (1) = \left. {\frac{\left( {{f}''_{\Psi d} (x)}
\right)^{\prime \prime }}{\left( {{f}''_{\Psi T} (x)} \right)^{\prime \prime
}}} \right|_{x = 1} = \frac{\textstyle{{33} \over 4}}{6} = \frac{11}{4}.
\end{equation}

By the application of the inequalities (\ref{eq5}) with (\ref{eq49}) we get the required result.
\end{proof}

\textbf{Proof of the Theorem 4.1.} We know that
\begin{equation}
\label{eq50}
D_{\Psi \Delta } (P\vert \vert Q) \le \frac{4}{3}D_{\Psi h} (P\vert \vert
Q).
\end{equation}

Propositions 4.1-4.8 together the inequality (\ref{eq50}) completes the proof of
the theorem.

\bigskip
From the sequence of inequalities given in theorem 4.1, we noted the absence
of the measure $D_{Jd} (P\vert \vert Q)$. Here below is an inequality
relating the measures $D_{Jd} (P\vert \vert Q)$ and $D_{Td} (P\vert \vert
Q)$.

\begin{proposition} We have
\[
\frac{1}{9}D_{Td} (P\vert \vert Q) \le D_{Jd} (P\vert \vert Q).
\]
\end{proposition}

\begin{proof} Let us consider
\[
g_{Td\_Jd} (x) = \frac{{f}''_{Td} (x)}{{f}''_{Jd} (x)} = \frac{2\left[
{4\sqrt x \left( {x + 1} \right)\left( {x^{3 / 2} + 1} \right) - \left( {x^2
+ 1} \right)\left( {2x + 2} \right)^{3 / 2}} \right]}{\left( {x + 1}
\right)\left[ {8\sqrt x \left( {x^{3 / 2} + 1} \right) - \left( {x + 1}
\right)\left( {2x + 2} \right)^{3 / 2}} \right]} , \quad
x \ne 1
\]

\noindent
for all $x \in (0,\infty )$, where ${f}''_{Td} (x)$ and ${f}''_{Jd} (x)$.
Calculating the first order derivative of the function $g_{Td\_Jd} (x)$ with
respect to $x$, one gets
\begin{equation}
\label{eq51}
{g}'_{Td\_Jd} (x) = - \frac{4\left( {\sqrt x - 1} \right)\left( {\sqrt x +
1} \right)\left( {x + 1} \right)\sqrt {2x + 2} }{\sqrt x \left( {x + 1}
\right)^2\left[ {8\sqrt x \left( {x^{3 / 2} + 1} \right) - \left( {x + 1}
\right)\left( {2x + 2} \right)^{3 / 2}} \right]^2}\times k_7 (x),
\end{equation}

\noindent
where $k_7 (x)$, $x > 0$ is given by
\[
k_7 (x) = 2\left( {x^{7 / 2} + 3x^{5 / 2} + 4x^2 + 4x^{3 / 2} + 3x + 1}
\right) - \sqrt x \left( {2x + 2} \right)^{5 / 2}.
\]

The graph of $k_7 (x)$, $x > 0$ is given by

\begin{figure}[htbp]
\centerline{\includegraphics[width=2.00in,height=2.00in]{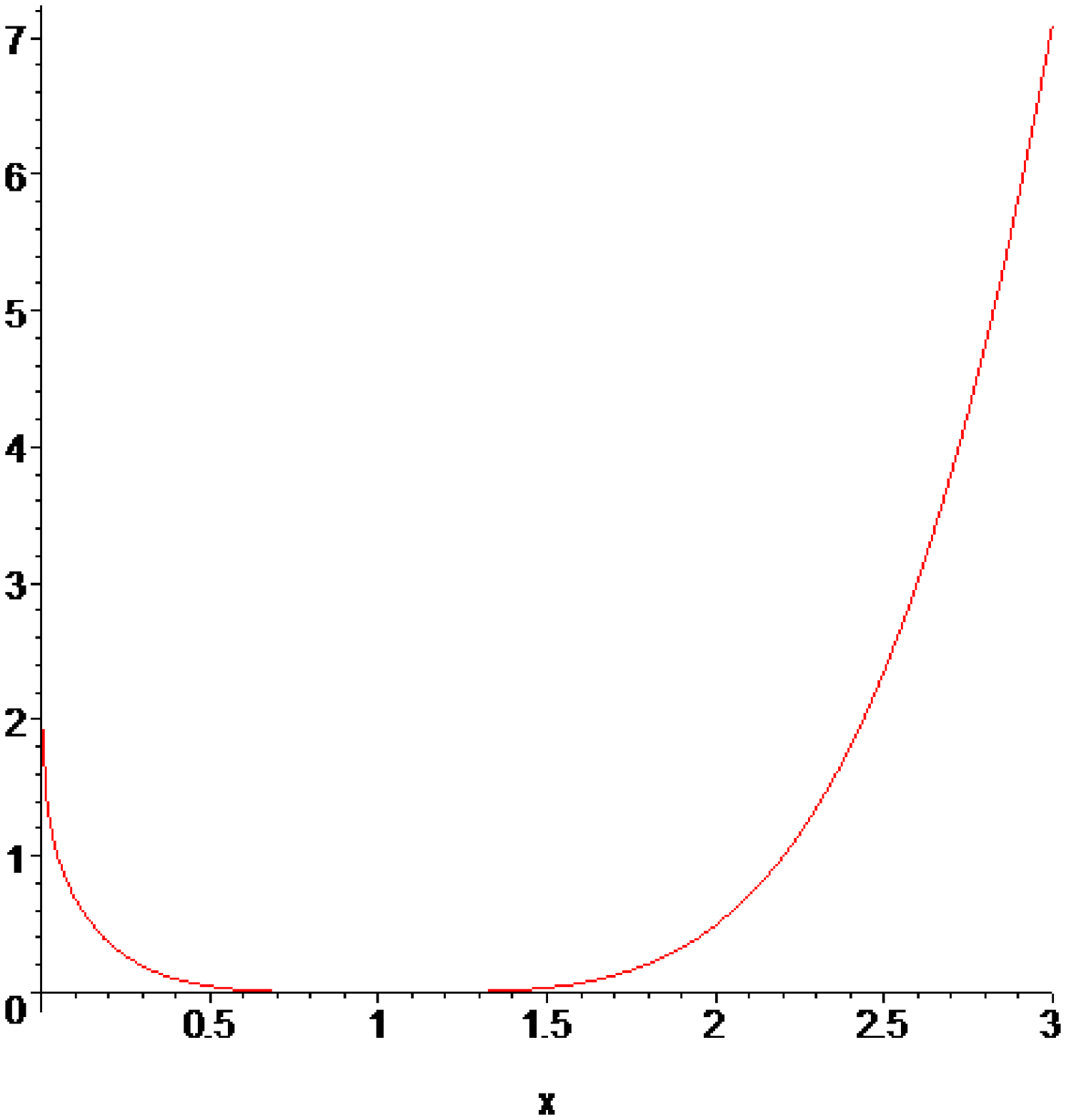}}
\label{fig10}
\end{figure}

We observe from the above graph that the function $k_7 (x) \ge 0$, $\forall
x > 0$. Thus from the from the expression (\ref{eq51}), we conclude that
\begin{equation}
\label{eq52}
{g}'_{Td\_Jd} (x) = \begin{cases}
 { > 0,} & {x < 1} \\
 { < 0,} & {x > 1} \\
\end{cases}
\end{equation}

Let us calculate now $g_{Td\_Jd} (1)$. We observe that
\[
\left. {g_{Td\_Jd} (x)} \right|_{x = 1} = \left. {\frac{{f}''_{Td}
(x)}{{f}''_{Jd} (x)}} \right|_{x = 1} = \left. {\frac{\left( {{f}''_{Td}
(x)} \right)^\prime }{\left( {{f}''_{Jd} (x)} \right)^\prime }} \right|_{x =
1} = \mbox{indermination}.
\]

Calculating the second order derivatives of numerator and denominator of the
function $g_{Td\_Jd} (x)$, we have
\begin{equation}
\label{eq53}
g_{Td\_Jd} (1) = \left. {\frac{\left( {{f}''_{Jd} (x)} \right)^{\prime
\prime }}{\left( {{f}''_{Td} (x)} \right)^{\prime \prime }}} \right|_{x = 1}
= \frac{\textstyle{9 \over 8}}{\textstyle{1 \over 8}} = 9.
\end{equation}

By the application of the inequalities (\ref{eq5}) with (\ref{eq53}) we get the required result.
\end{proof}

\section{Final Remarks}

\begin{remark} In view of (\ref{eq23}) and Proposition 4.10, we have
\begin{align}
\label{eq54}
D_{h\Delta } (P\vert \vert Q) \le & \frac{4}{5}D_{d\Delta } (P\vert \vert Q)
\le 4D_{dh} (P\vert \vert Q) \le \frac{12}{7}D_{dI} (P\vert \vert Q) \le\notag\\
& \le 3D_{hI} (P\vert \vert Q) \le D_{Th} (P\vert \vert Q) \le
\frac{4}{3}D_{Td} (P\vert \vert Q) \le 12\,D_{Jd} (P\vert \vert Q).
\end{align}

We can easily find examples where the measure $D_{Jd} (P\vert \vert Q)$ don't have relations with the other measures appearing the rest of sequence given in (\ref{eq23}), such as $D_{\Psi \Delta } (P\vert \vert Q)$, etc.
\end{remark}

\begin{remark} Following the similar lines of above propositions, we can also prove the following inequality for the measure $D_{Jd} (P\vert \vert Q)$
\begin{equation}
\label{eq55}
\frac{1}{4}D_{Jh} (P\vert \vert Q) \le D_{Jd} (P\vert \vert Q).
\end{equation}
\end{remark}

\begin{remark} As a consequence of Propositions 4.1-4.9 and the expression (\ref{eq23}), we have the following inequalities:
\begin{itemize}
\item[(i)] \quad $\frac{16d+3I}{7} \le h\le \frac{64d+3\Delta }{20} $.
\item[(ii)] \quad $h\le \frac{T+3I}{4} $.
\item[(iii)] \quad $h\le \frac{\Psi +12\Delta }{64} $.
\item[(iv)] \quad $4d\le \frac{T+3h}{4} $.
\item[(v)] \quad $4d\le \frac{\Psi +176h}{192} $.
\item[(vi)] \quad $4d\le \frac{3J+8h}{8} $.
\item[(vii)] \quad $T\le \frac{3\Psi +512d}{176} $.
\item[(viii)] \quad $\frac{32d+T}{9} \le \frac{J}{8} $.
\item[(ix)] \quad $4T+\frac{3\Delta }{16} \le \frac{3\Psi }{64} +16d$.
\end{itemize}

In view of above above results we have the following inequalities:
\begin{align}
& I\le \frac{16d+3I}{7} \le h \le \frac{64d+\Delta }{20} \le 4d\le \frac{32d+T}{9} \le \frac{1}{8} J,  \notag\\
 &h\le \frac{\Psi +12\Delta }{64} \le 4d\le \frac{T+3h}{4} \le T\le 4d+\frac{3}{16} \left(\frac{1}{16} \Psi -\frac{1}{4} \Delta \right)\le \frac{3J+8h}{8} ,\notag\\
& 4d\le \frac{\Psi +176h}{192} \le \frac{1}{16} \Psi ,\notag
\intertext{and}
&T\le \frac{3\Psi +512d}{176} \le \frac{1}{16} \Psi . \notag
\end{align}
\end{remark}

\begin{remark} As a consequence of previous results (\ref{eq2}) and (\ref{eq3}), we have the following inequalities:
\begin{align}
&\frac{1}{4} \Delta \le I\le \frac{9h+\Delta }{12} \le h\le \frac{J+8I}{16} \le \frac{T+2h}{3} \le \frac{1}{8} J\le \frac{8T+\Delta }{12} \le T\le \frac{\Psi +6J}{642} \le \frac{1}{16} \Psi \notag\\
\intertext{and}
&\frac{1}{8} J\le \frac{\Psi +192h-4\Delta }{192}  \le \frac{\Psi +128h}{144} \le \frac{1}{16} \Psi .\notag \end{align}
\end{remark}

\end{document}